\documentclass[sigconf, nonacm]{acmart}
\newcommand\vldbdoi{XX.XX/XXX.XX}

\usepackage{graphicx}
\usepackage{multirow}
\usepackage{algorithm}
\usepackage{tikz}
\usepackage[english]{babel}
\usepackage[noend]{algpseudocode}
 \usepackage{float}
\usepackage{multirow}
\usepackage{subfigure}
\usepackage{pgfplots}
\usepackage{color, colortbl}
\usepackage{varwidth}
\usepackage{hyperref}
\usepackage{balance}
\definecolor{darkred}{rgb}{0.55, 0.0, 0.0}
\definecolor{ao}{rgb}{0.0, 0.0, 1.0}

\tikzset{
  vertex/.style={circle,draw}
}

\pgfplotsset{compat=newest} 

\newcommand{\Sedge}{\mbox{Sedge}}
\newcommand{\nSedge}{\mbox{nSedge}}

\begin{document}
\title{Utility-Based Graph Summarization: New and Improved}

\settopmatter{authorsperrow=4}
\author{Mahdi Hajiabadi}
\affiliation{%
  \institution{University of Victoria}
  \streetaddress{P.O. Box 1212}
  \city{Victoria}
  \state{Canada}
  \postcode{43017-6221}
}
\email{mhajiabadi@uvic.ca}

\author{Jasbir Singh}
\orcid{0000-0002-1825-0097}
\affiliation{%
  \institution{University of Victoria}
  \streetaddress{P.O. Box 1212}
  \city{Victoria}
  \state{Canada}
  \postcode{43017-6221}
}
\email{jasbirsingh@uvic.ca}

\author{Venkatesh Srinivasan}
\orcid{0000-0001-5109-3700}
\affiliation{%
  \institution{University of Victoria}
  \streetaddress{P.O. Box 1212}
  \city{Victoria}
  \state{Canada}
  \postcode{43017-6221}
}
\email{srinivas@uvic.ca}

\author{Alex Thomo}
\affiliation{%
  \institution{University of Victoria}
  \streetaddress{P.O. Box 1212}
  \city{Victoria}
  \state{Canada}
  \postcode{43017-6221}
}
\email{thomo@uvic.ca}





\begin{abstract}
\looseness=-1
A fundamental challenge in graph mining is the ever increasing size of datasets. 
Graph summarization aims to find a compact representation resulting in faster algorithms and reduced storage needs. 
The flip side of graph summarization is often loss of utility which significantly diminishes its usability. 
The key questions we address in this paper are:
{\em 
(1)~How to summarize a graph without any loss of utility?
(2)~How to summarize a graph with some loss of utility but above a user-specified threshold?
(3)~How to query graph summaries without graph reconstruction?} 
We also aim at making graph summarization available for the masses by efficiently handling web-scale graphs using only a consumer grade machine.
Previous works suffer from conceptual limitations and lack of scalability. 

In this work, we make three key contributions. 
First, we present a utility-driven graph summarization method, based on a clique and independent set decomposition, that produces significant compression with zero loss of utility. The compression provided is significantly better than state-of-the-art in lossless graph summarization, while the runtime is two orders of magnitude lower. 
Second, we present a highly scalable algorithm for the lossy case, which foregoes the expensive iterative process that hampers previous work. Our algorithm achieves this by combining a memory reduction technique and a novel binary-search approach.   
In contrast to the competition, we are able to handle web-scale graphs in a single machine without performance impediment as the utility threshold (and size of summary) decreases.
Third, we show that our graph summaries can be used as-is to answer several important classes of queries, such as triangle enumeration, Pagerank, and shortest paths. 
This is in contrast to other works that incrementally reconstruct the original graph for answering queries, thus incurring additional time costs.
\end{abstract}

\maketitle

 \section{Introduction}

Graphs are ubiquitous and are the most natural representation for many real-world data such as web graphs, social networks, communication networks, citation networks, transaction networks, ecological networks and epidemiological networks. Such graphs are growing at an unprecedented rate. 
For instance, 
the web graph consists of more than a trillion websites \cite{websites} and 
the social graphs of Facebook, Twitter, and Weibo, have billions of users with many friend/follow connections per user~\cite{facebook,twitter,weibo}. 
%
%
Consequently, storing such graphs and answering queries, mining patterns, and visualizing them are becoming highly impractical \cite{1,5}. 
\par
Graph summarization is a fundamental task of finding a compact representation of the original graph called the summary. 
It allows us to decrease the footprint of the graph and query more efficiently \cite{3,4,5}.
Graph summarization also makes possible effective visualization thus facilitating better insights on large-scale graphs \cite{shen2006visualanalysis,19,20,21,22}. 
Also crucial is the privacy that a graph summary can provide for privacy-aware graph analytics \cite{2,23}.

The problem has been approached from different directions, such as compression techniques to reduce the number of required bits for describing graphs \cite{rossi2018graphzip, apostolico2009graph,boldi2004webgraph,shah2017summarizing}, sparsification techniques to remove less important nodes/edges in order to make the graph more informative \cite{spielman2011graph, 9} and 
grouping methods that merge nodes into supernodes based on some interestingness measure \cite{2,4,5,6,24}. 
Grouping methods constitute the most popular summarization approach because they allow the user to logically relate the graph summary to the original graph.  

The flip side of summarization is loss of utility. 
This is measured in terms of edges of the original graph that are lost and spurious edges that are introduced in the summary. 
In this paper, we focus on grouping-based utility-driven  graph summarization. 
In terms of state-of-the-art, \cite{5} and \cite{1} offer different ways of measuring loss of utility. The first computes the loss by assuming all edges as unweighted and of equal importance while the second incorporates edge centralities as weights in the loss computation. Also, the first uses a loss budget that is local to each node while the second uses a global budget.  

There are several limitations with state-of-the-art \cite{5,1} on utility-driven graph summarization. 
By not considering edge importance, the SWeG algorithm of \cite{5} produces (lossy) summaries which are inferior to those produced by UDS of \cite{1} which uses edge importance in its process. 
UDS, however, is not able to generate a meaningful lossless summary if losslessness is required by the application. SWeG, on the other hand, has the option to produce lossless summaries, but extending SWeG to use edge importance for the lossy case is not trivial.
Both SWeG and UDS are slow and impractical to run for large datasets in a single machine, thus hampering their utility. 
SWeG needs to utilize a cluster of machines to be able to handle datasets that are large but still can fit easily in the memory of one machine. 
UDS is a $O(V^2)$-time algorithm; based on our experiments, it can only handle small to moderate datasets requiring a large amount of time, often more than 100 hours. 

To address these challenges, we propose two utility-driven algorithms, G-SCIS and T-BUDS, for the lossless and lossy cases, respectively, which can handle large graphs efficiently on a single consumer-grade machine. 
G-SCIS is based on a clique and independent set decomposition that produces significant compression with zero loss of utility. Compared to SWeG, G-SCIS produces much better summaries with respect to reduction in number of nodes, while having a runtime which is lower by two-orders of magnitude.
We reiterate here that UDS is not able to produce a lossless summary that is different from the original graph, and as such, is not a contender in the lossless case.

We also show that G-SCIS summaries possess an attractive characteristic not present in SWeG (or other methods) summaries. 
Due to our clique and independent set decomposition, we are able to compute important classes of queries, such as Pagerank, triangle enumeration, and shortest paths using the G-SCIS summary ``as-is'' without the need to perform a reconstruction of the original graph. In contrast, for SWeG summaries, we need to use neighborhood queries as primitives, which amounts to incrementally reconstructing the original graph, thus incurring additional time costs. 

Our second algorithm, T-BUDS, is a highly scalable algorithm for the lossy case. 
It shares the utility-threshold-driven nature of UDS \cite{1} allowing the user to calibrate the loss of utility according to their needs. However, T-BUDS forgoes the main expensive iterative process that severely hampers UDS. We achieve this by combining a memory reduction technique based on Minimum Spanning Tree and a novel binary-search approach. 
T-BUDS not only is orders of magnitude faster than UDS, but it also exhibits a useful characteristic; namely, T-BUDS (in contrast to UDS) is mostly computationally insensitive to lowering the utility threshold which amounts to asking for smaller size summary. As such, a user can conveniently experiment with different utility thresholds without incurring a performance penalty.


In summary, our contributions are as follows.

\begin{itemize}
    \item We propose an optimal algorithm, G-SCIS, for lossless graph summarization  and show that it outperforms the state of art, SWeG, by two orders of magnitude in runtime while achieving better reduction in number of nodes. 
    
    \item We show interesting applications of the summary produced by G-SCIS to triangle enumeration, Pagerank, and shortest path queries. 
    For instance, we show that we can enumerate triangles and compute Pagerank on the G-SCIS summaries much faster than on the original graph. 
    
	\item We propose a highly scalable, utility-driven algorithm, T-BUDS, for lossy summarization. This algorithm achieves high scalability by combining memory reduction using MST with a novel binary search procedure. While preserving all the nice properties of the UDS summary, T-BUDS outperforms UDS by two orders of magnitude.
	
	
	
	\item We show that the summary produced by T-BUDS,  can be used to answer top-$k$ central node queries based on various centrality measures with high level of accuracy.
	
\end{itemize}


\section{Preliminaries \label{sec:Preliminary}}
Let $G=(V,E)$ be an undirected graph, where 
$V$ is the set of nodes and 
$E$ the set of edges.  
A summary graph is also undirected and denoted by 
$\mathcal{G}= (\mathcal{V,E})$, where  $\mathcal{V}$ is a set of supernodes, and
$\mathcal{E}$ a set of superedges. 

More precisely we have 
$\mathcal{V}=\{S_1,S_2, \dots, S_k\}$
such that $k\leq |V|$, $V = \bigcup_{i=1}^{k}S_i$ and $\forall i\neq j, S_i\cap S_j = \emptyset$. 
The supernode which a node $u\in V$ belongs to is denoted by $S(u)$. 

\smallskip
\noindent
{\bf Reconstruction.} 
Given a summary graph, we can (lossily) reconstruct the original graph as follows. 
For each superedge $(S_i,S_j)$ we construct edges $(u,v)$, for each $u\in S_i$ and $v \in S_j$.
For $i\neq j$, 
this amounts to building a complete bipartite graph with $S_i$ and $S_j$ as its {\em parts}.
For $i = j$ (a self-loop superedge), 
the reconstruction amounts to building a clique among the vertices of $S_i$.
Figure \ref{fig:differentsuperedge} shows how the reconstructed graph is affected by different types of superedges. Figures (a) and (c) show two different superedges and figures (b) and (d) show their reconstructed versions.
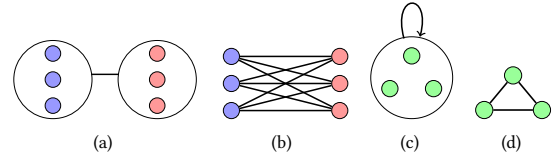
\begin{figure}
\subfigure[]{
\resizebox{.14\textwidth}{!}{
\begin{tikzpicture}[node distance = {0.25cm and .15cm}, every node/.style = {draw, circle}]
    \node[fill = blue!40] (A1) at (0.5,0) {};
    \node[fill = blue!40] (B1) at (0.5,-0.5) {};
    \node[fill = blue!40] (C1) at (0.5,-1) {};
    \node[draw=none,scale=4.7] (D) at (0.5,-0.4){};
    \draw [](0.5,-0.5) circle (0.75cm);
    \node[fill = red!40] (A2) at (2.5,0) {};
    \node[fill = red!40] (B2) at (2.5,-0.5) {};
    \node[fill = red!40] (C2) at (2.5,-1) {};
    \draw [](2.5,-0.5) circle (0.75cm);
    \node[draw=none,scale=4.7] (D1) at (2.5,-0.4){};
     \draw[thick] (D) --  (D1);
 \end{tikzpicture}
 }
 }
 \subfigure[]{
    \resizebox{.1\textwidth}{!}{
        \begin{tikzpicture}[node distance = {0.25cm and .15cm}, every node/.style = {draw, circle}]
            \node[fill = blue!40] (A1) at (0.5,-1.5) {};
            \node[fill = blue!40] (B1) at (0.5,-2) {};
            \node[fill = blue!40] (C1) at (0.5,-1) {};
            \node[fill = red!40] (A2) at (2.5,-1.5) {};
            \node[fill = red!40] (B2) at (2.5,-2) {};
            \node[fill = red!40] (C2) at (2.5,-1) {};
            \draw[thick] (A1) --  (A2);
            \draw[thick] (A1) --  (B2);
            \draw[thick] (A1) --  (C2);
            \draw[thick] (B1) --  (A2);
            \draw[thick] (B1) --  (B2);
            \draw[thick] (B1) --  (C2);
            \draw[thick] (C1) --  (A2);
            \draw[thick] (C1) --  (B2);
            \draw[thick] (C1) --  (C2);
        \end{tikzpicture}
    }
 }
 \subfigure[]{
    \resizebox{.065\textwidth}{!}{
        \begin{tikzpicture}[node distance = {0.25cm and .15cm}, every node/.style = {draw, circle}]
            \node[fill = green!40] (A1) at (5.5,0) {};
            \node[fill = green!40] (B1) at (5.1,-0.6) {};
            \node[fill = green!40] (C1) at (5.9,-0.6) {};
            \draw [](5.5,-0.4) circle (0.75cm);
            \node[draw=none,scale=3.5] (D2) at (5.5,-0.21){};
             \draw[thick] (D2) edge [loop above ] (D2);
        \end{tikzpicture}
    }
}
 \subfigure[]{
    \resizebox{.06\textwidth}{!}{
        \begin{tikzpicture}[node distance = {0.25cm and .15cm}, every node/.style = {draw, circle}]
            \node[fill = green!40] (A3) at (5.5,-1.4) {};
            \node[fill = green!40] (B3) at (5,-2) {};
            \node[fill = green!40] (C3) at (6,-2) {};
            \draw[thick] (A3) --  (B3);
            \draw[thick] (A3) --  (C3);
            \draw[thick] (B3) --  (C3);
        \end{tikzpicture}
    }
}
\caption{(a,c) Two different type of superedges which result in two different types of reconstructed graph (b,d). \label{fig:differentsuperedge}}
\end{figure}

\smallskip
\noindent
{\bf Utility.}
In order to reason about the utility of a graph summarization we need to define the notion of edge importance. We denote the importance of an edge $(u,v)$ in $G$ by $C(u,v)$.
For example, the edge importance could measure its centrality. 
Obviously, the more important edges we recover during reconstruction, the better it is.
However, this should not come at the cost of introducing spurious edges. 
In order to measure the amount of spuriousness, we also introduce the notion of importance for spurious edges and denote that by $C_s(u,v)$.
Now we give the definition of utility as follows. 
\begin{equation}
    u(\mathcal{G}) = \sum_{\substack{(S_i,S_j)\in \mathcal{E}}}\left(\sum_{\substack{(u,v)\in E \\ u \in S_i, v \in S_j}} C(u,v) - \sum_{\substack{(u,v)\notin E \\ u \in S_i, v \in S_j} } C_s(u,v)\right)
    \label{eq:utilitythreshold}
\end{equation}
In order to have a good summarization, the user defines a threshold $\tau$ and requests that 
$u(\mathcal{G})>\tau$. 
The $C(u,v)$ and $C_s(u,v)$ values are normalized so that their respective sums equal one. 
A similar utility model is used in \cite{1} but without weights for spurious edges. 

Figure \ref{fig:exampleUtility} shows an example for this framework. There are 14 edges and 11 nodes.
We assume the weight of each actual edge is equal to $\frac{1}{|E|} = \frac{1}{14}$ and the weight of each spurious edge is equal to $\frac{1}{\binom{11}{2} - 14} = \frac{1}{41}$. 
In part (a) the set of nodes inside the circles merge together into new supernodes and the utility still remains one because no information has been lost. 
In part (b) the circles show two merge cases. 
In the first case, the blue supernode merges with the red node and in the second case, the green supernode merges with the blue node. In the first case, there is a utility loss of $\frac{1}{14}$ for missing one actual edge (see part (d) for the reconstructed graph). We chose not to add an edge from the new blue supernode to one of the neighbours of the red node because doing so would introduce three spurious edges for a cost of $\frac{3}{41}$ that is greater than $\frac{1}{14}$ (cost of missing one actual edge). Similarly, in the second case, there is a utility loss of $\frac{2}{41}$ for introducing two spurious edges. Therefore, the utility after this step is  $1 - \frac{1}{14} - \frac{2}{41} = \frac{505}{574}$. 
Part (c) shows the summary after all the four merges and part (d) shows the reconstructed graph of summary in part (c). 
\par
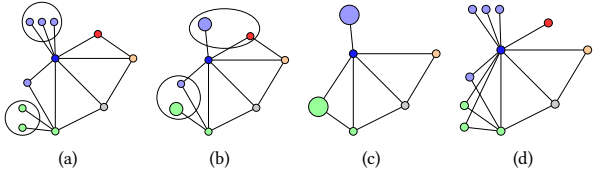
\begin{figure}
\subfigure[]{
   \resizebox{.1\textwidth}{!}{
\begin{tikzpicture}[node distance = {1.0cm and .5cm}, every node/.style = {draw, circle}]
    \node[fill = green!40] (A) at (0,0) {};
    \node[fill = green!40] (A1) at (-1.35,0.95) {};
    \node[fill = green!40] (A2) at (-1.35,0.15) {};
    \draw [](-1.35,0.55) circle (0.71cm);
    \node[fill = blue!90] (B) at (0,3) {};
    \node[fill = blue!40] (B1) at (-0.05,4.5) {};
    \node[fill = blue!40] (E) at (-0.55,4.5) {};
    \node[fill = blue!40] (B2) at (-1.15,2) {};
    \node[fill = blue!40] (B3) at (-1.05,4.5) {};
    \draw [](-0.55,4.5) circle (0.81cm);
    \node[fill = red!80] (C) at (1.75,4) {};
    \node[fill = gray!40] (D) at (2,1) {};
    \node[fill = orange!40] (F) at (3.2,3) {};
    \draw (A) --  (B);
    \draw (B) -- (B2);
    \draw (B) -- (B3);
    \draw (B) -- (B1);
    \draw (A) -- (B2);
    \draw (A) --  (A1);
    \draw (A) --  (A2);
    \draw  (B) -- (C);
    \draw  (B) -- (F);
    \draw (D) -- (A);
    \draw (D) -- (B);
    \draw (D) -- (F);
    \draw  (C) -- (F);
    \draw  (B) -- (E); 
\end{tikzpicture}
}
}
\subfigure[]{
\resizebox{.1\textwidth}{!}{
\begin{tikzpicture}[node distance = {1.0cm and .5cm}, every node/.style = {draw, circle}]
\node[fill = green!40] (A) at (0,0) {};
    \node[fill = blue!90] (B) at (0,3) {};
    \node[fill = blue!40] (B2) at (-1.15,2) {};
    \node[fill = red!80] (C) at (1.75,4) {};
    \node[fill = gray!40] (D) at (2,1) {};
    \node[fill = orange!40] (F) at (3.2,3) {};
    \node[fill = green!40,minimum size=0.55cm] (A1) at (-1.35,.95) {};
    \draw [](-1.25,1.43) circle (0.91cm);
    \node[fill = blue!40,minimum size=0.55cm] (B1) at (-0.15,4.5) {};
    \draw (0.67,4.32) ellipse (1.47cm and 0.83cm);
    \draw (A) --  (B);
    \draw (B) -- (B2);
    \draw (B) -- (F);
    \draw (B) -- (B1);
    \draw (A) -- (B2);
    
     \draw (A) --  (A1);
    \draw  (B) -- (C);
    \draw (D) -- (A);
    \draw (D) -- (B);
    \draw (D) -- (F);
    \draw  (C) -- (F);
\end{tikzpicture}
}
}
\subfigure[]{
\resizebox{.1\textwidth}{!}{
\begin{tikzpicture}[node distance = {1.0cm and .5cm}, every node/.style = {draw, circle}]
    \node[fill = green!40] (A) at (0,0) {};
    \node[fill = blue!90] (B) at (0,3) {};
    \node[fill = gray!40] (D) at (2,1) {};
    \node[fill = orange!40] (F) at (3.2,3) {};
    \node[fill = green!40,minimum size=0.75cm] (A1) at (-1.35,.95) {};
    \node[fill = blue!40,minimum size=0.75cm] (B1) at (-0.15,4.5) {};
    
    \draw (A) --  (B);
    \draw (B) -- (A1);
    \draw (B) -- (B1);
    \draw (A) -- (D);
     \draw (A) --  (A1);
    \draw  (B) -- (F);
    \draw (D) -- (B);
    \draw (D) -- (F);
\end{tikzpicture}
}
}
\subfigure[]{
\resizebox{.1\textwidth}{!}{
\begin{tikzpicture}[node distance = {1.0cm and .5cm}, every node/.style = {draw, circle}]
     \node[fill = green!40] (A) at (0,0) {};
    \node[fill = green!40] (A1) at (-1.35,0.95) {};
    \node[fill = green!40] (A2) at (-1.35,0.15) {};
    \node[fill = blue!90] (B) at (0,3) {};
    \node[fill = blue!40] (B1) at (-0.05,4.5) {};
    \node[fill = blue!40] (E) at (-0.55,4.5) {};
    \node[fill = blue!40] (B2) at (-1.15,2) {};
    \node[fill = blue!40] (B3) at (-1.05,4.5) {};
    \node[fill = red!80] (C) at (1.75,4) {};
    \node[fill = gray!40] (D) at (2,1) {};
    \node[fill = orange!40] (F) at (3.2,3) {};
    
    \draw (A) --  (B);
    \draw (A1) --  (B);
    \draw (A2) --  (B);
    \draw (B) -- (B2);
    \draw (B) -- (B3);
    \draw (B) -- (B1);
    \draw (A) -- (B2);
    \draw (A) --  (A1);
    \draw (A) --  (A2);
    \draw  (B) -- (C);
    \draw  (B) -- (F);
    \draw (D) -- (A);
    \draw (D) -- (B);
    \draw (D) -- (F);
    \draw  (B) -- (E); 
\end{tikzpicture}
}
}
\caption{Example of the utility-based framework. 
(a) Shows the original graph with two candidate merges with no loss of utility. The result is shown in (b) along with two more candidate merges. 
The merge of the green supernode with the blue node introduces two spurious edges (see the relevant part in the reconstructed graph in (d)). 
The merge of the blue supernode with the red node loses an actual edge as shown by the result in (d). (d) shows the reconstructed graph starting from the summary graph in (c).
\label{fig:exampleUtility}}
\end{figure}

\begin{table}[]
    \begin{tabular}{|m{0.125\textwidth} |m{0.31\textwidth}|}  \hline
        Symbols & Definition \\ \hline
        $G = (V,E)$ &  Input graph with set $V$ of nodes and set $E$ of edges \\  \hline
         $\mathcal{G} = (\mathcal{V}, \mathcal{E})$ & Summary graph with set $\mathcal{V}$ of supernodes and set $\mathcal{E}$ of superedges  \\ \hline
         $S_i$ & The $i$-th supernode  \\ \hline
         $S(u)$ & Supernode to which node $u$ belongs \\ \hline         $C_u$ & Centrality of node $u$ \\ \hline
         $C(u,v)$ & Centrality of edge $(u,v)$ \\ \hline
         $F$ & $\{(a,c)|(a,b)\in E , (b,c) \in E\}$ \\ \hline
         $G_{2-hop}  =(V,F)$ &  2-hop graph of $G$ \\ \hline
         $L$ & Sorted list of $F$  \\ \hline
         $H$ & Sorted list of MST of $G_{2-hop}$ \\ \hline
         $d_{avg}$ & Average degree \\ \hline
         $u(\mathcal{G}^t)$ & The utility value of the summary after iteration $t$ \\ \hline
         $\mbox{Sedge}(S(u),S(v))$ & Cost of adding a superedge  between $S(u)$ and $S(v)$\\ \hline
         $\mbox{nSedge}(S(u), S(v))$ & Cost of not adding a  superedge  between $S(u)$ and $S(v)$ \\ \hline
         $N(u)$ &  The neighborhood set of $u$ in graph $G$ \\ \hline
         $N(X)$ & The neighborhood set of $X$ in graph $\mathcal{G}$ \\ \hline
         $|X|$ & The number of nodes in supernode $X$ \\ \hline
         $nodes(X)$ & The set of all nodes in supernode $X$ \\ \hline
    \end{tabular}
    \caption{Table of frequently used symbols}
    \vspace*{-1cm}
    \label{tbl:notations}
\end{table}
Having described the utility-based framework, we  define the optimization problem we study as follows.
Given graph $G=(V,E)$ and user-specified utility threshold $\tau$, our objective is to
\begin{equation}
\begin{split}
    \mbox{minimize } |\mathcal{V}|  \\
    \mbox{ subject to $u(\mathcal{G}) \geq \tau$}.
    \end{split}
\end{equation}

\section{Optimal lossless algorithm \label{sec:OptLessless}}

\looseness=-1
Kumar et al. \cite{1} showed that given a general utility threshold~$\tau$, graph summarization is \emph{NP-Hard}. 
Furthermore, for any $\epsilon > 0$, there is no efficient $O(n^{1-\epsilon})$-approximation algorithm for finding the minimum number of supernodes.  
In this section, we analyze the problem for the special case of $\tau = 1$, that is, lossless graph summarization.
When we reconstruct the graph from such a summary, no actual edge will be lost and no spurious edge will be introduced.
We remark that the UDS algorithm of \cite{1} is such that the summary obtained for $\tau=1$ is in fact just the original graph. In this section we show that we can do much better, namely we obtain in {\em polynomial time} the {\em optimal} summary in terms of the objective function, i.e. we obtain the summary with the smallest number of supernodes.
We start with the following lemma.

\begin{lemma} 
    In any summary corresponding to $\tau =1$,  each node can be  
(1) in a supernode of size one, or 
(2) inside a supernode representing a clique in $G$ with size greater than one, or 
(3) inside a supernode representing an independent set in $G$ (a set of nodes where no two nodes are connected) of size greater than one.
\end{lemma}

\begin{proof}
Recall that during reconstruction, 
a supernode either generates just 
one node (when there is only one node in the supernode), 
a clique (when a self-loop exists), or 
an independent set (when a self-loop does not exist).
Now if the original graph does not precisely correspond to what is reconstructed, then there will be at least either one spurious edge added (in the case of a clique supernode), or one actual edge lost (in the case of an independent set supernode). Thus the summary would not be lossless and the reconstructed graph would be different from the original graph.
\end{proof}

We observe another property of lossless summaries.

\begin{lemma} \label{lemma:either-clique-or-is}
A node $v$ cannot be in a clique supernode in one lossless summary and in an independent set supernode in another.
\end{lemma}
\begin{proof}
 For a contradiction, let us assume that nodes $v_i,v_j$ are inside an independent set supernode in one lossless summary and $v_i,v_k$ are inside a clique supernode in another. This implies that  $N(v_i)$ (set of neighbors of $v_i$) is exactly the same as  $N(v_j)$. Since $v_i,v_k$ are inside a clique supernode, $v_k \in N(v_i)$ and thus $v_k \in N(v_j)$ (because  $N(v_i) = N(v_j)$). Also, since $v_i,v_j$ are inside an independent set, $v_i \notin N(v_j)$.  Further, for $v_i$ and $v_k$ to be in the same clique supernode, $N(v_i) \setminus \{v_i\}$ should be same as $N(v_k) \setminus \{v_k\}$ but this is violated as $v_j$ is connected to $v_k$ and not to $v_i$. Hence the contradiction.
\end{proof}

We now show that there is a polynomial-time algorithm that computes the optimal lossless summarization. 
%
Algorithm~\ref{alg:tau=1} given below proposes a global greedy strategy for finding the optimal summary. For each node $u$, the goal of  the algorithm is to find the biggest supernode that $u$ can be a part of. For the summary to be lossless, such a supernode has to be either an independent set or a clique.

\begin{algorithm}[]
    \caption{Finding the best summary for $\tau=1$}\label{alg:tau=1}
    \begin{algorithmic}[1]
        \State \textbf{Input:} $G = (V,E)$
        \State \textbf{Initialization: }{$Status[\forall v \in V] \gets \mbox{\sl False}$, $\mathcal{S}$ $\gets$ []}
       
        \For{$u \in V \land Status[u] = \mbox{\sl False}$}
            \State{$S(u) \gets \{u\}$}
            \State{$Status[u] \gets \mbox{\sl True}$}
            \For{$v \in V \land Status[v] = \mbox{\sl False}$}
                \If{$(N(u) = N(v)) \lor (N(u)\setminus \{v\} = N(v) \setminus \{u\})$} \label{condition}
                    \State{$S(u) \gets S(u) \cup \{v\}$}
                    \State{$Status[v] \gets \mbox{\sl True} $}
                \EndIf
            \EndFor
            \State{$\mathcal{S}.add(S(u))$}
        \EndFor
         \State{$\Call{BuildSuperEdges}{\mathcal{S}}$}
    \end{algorithmic}
\end{algorithm}

Condition $(N(u) = N(v))$ in line \ref{condition} of Algorithm \ref{alg:tau=1} states that, if nodes $u,v$ share the same neighborhood set, then they are part of an independent set and should be merged. 
Condition $(N(u)\setminus \{v\} = N(v) \setminus \{u\})$ in line \ref{condition} states that, if $u,v$ are connected by an edge and they share the same neighborhood set, if we exclude $u$ from $N(v)$ and $v$ from $N(u)$, then they are part of a clique and should be merged. 
Further, Lemma~\ref{lemma:either-clique-or-is} proved that these conditions are mutually exclusive. If none of these conditions holds true, then node $u$ should be in a supernode of size one.

\smallskip
\noindent
{\bf Building Superedges.}
Once the appropriate supernodes have been identified
we build superedges as follows. 
For each supernode $S$, an edge is added to another supernode $S'$ iff $u\in S$ and $v \in S'$ and $(u,v)\in E$. We refer to this process as BuildSuperEdges (last line of Algorithm~\ref{alg:tau=1}).

\begin{theorem}[\textbf{Tractability of lossless graph summarization}] Lossless graph summarization is in $P$. That is, Algorithm~\ref{alg:tau=1} computes the optimal solution in polynomial time for $\tau = 1$. 
\end{theorem}
\begin{proof}

We claim that the supernode corresponding to any vertex $u \in V$ in the summary provided by Algorithm~\ref{alg:tau=1}, is the largest possible supernode for $u$ in any lossless summary. Suppose that $u$ is in an independent set supernode. All the other nodes inside that supernode must have the same neighbor set as $u$. Algorithm~\ref{alg:tau=1} greedily finds and adds all possible vertices $v \in V$ that have same neighborhood set as $u$ to the supernode. Hence, this must be the largest size possible.
An analogous argument applies for the case when $u$ is in a clique supernode. 

We now show that the Algorithm~\ref{alg:tau=1} produces an optimal lossless summary. For contradiction, let us assume that there exists an optimal lossless summary in which the number of supernodes is less than than the summary provided by Algorithm~\ref{alg:tau=1}. If so, there should exist at least one node $u \in V$ such that its supernode size in the optimal summary is larger than the its supernode size in the summary provided by Algorithm~\ref{alg:tau=1}. However, we proved in the previous paragraph that this can never happen and hence is a contradiction. 
Finally, it can be verified that 
the time complexity of Algorithm~\ref{alg:tau=1} is $O(V^2\Delta_{max})$, 
where $\Delta_{max}$ is the maximum degree of a node in $G$ and hence lossless summarization is in $P$.
\end{proof}

\subsection{Scalable Lossless Algorithm, G-SCIS}

Algorithm~\ref{alg:tau=1} is of $O(V^2 \Delta_{max})$ time complexity, which makes it impractical for large datasets. 
In this section, we propose an improved algorithm of $O(E)$ complexity, which uses hashing to speed up the process. 
We can break down the process into three parts: (a) finding candidate supernodes, (b) filtering supernodes, and (c) connecting superedges.  
A hash function is used to map each sorted neighbor set of the original graph into a number and all the nodes whose neighbor sets have same hash value are grouped into same candidate supernode.  

Note that the use of a hash function could result in candidate supernodes with false positives (i.e. two nodes that should not belong to same supernode might be present into one candidate supernode) but there cannot be false negatives (i.e. two nodes that must belong to same supernode cannot be in two different candidate supernodes). 
Of course, the probability of a false positive depends on the quality of the hash function used. 
In order to remove false positives, we further examine each candidate supernode for false positives, which are then filtered out into separate supernodes.  After this step all the supernodes are as they should be in an optimal summary and finally the superedges are added between them. 

In Algorithm \ref{alg:candidate-supernodes}, two different hash values ($h_c$ and $h_{i}$) are generated for the neighbor sets of each node. 
The nodes that have the same $h_c$ value (line \ref{line:hash-clique}) are grouped together to form candidate clique supernodes. 
Similarly, the nodes that have the same $h_{i}$ value (line \ref{line:hash-is}) are grouped together to form candidate independent set supernodes. 
Note that due to possible false positives, there can exist a node that is present in both a candidate independent set and a candidate clique at the same time. 
Finally, Algorithm~\ref{alg:candidate-supernodes} returns two hashmaps, $mapC$ and $mapI$,  where keys are hash values and buckets contain the set of nodes falling in the same candidate clique or independent set supernode.

Algorithm~\ref{alg:filter-supernodes} filters the candidate supernodes to become correct supernodes. 
For any candidate supernode, it selects a random node $u$, and, using its neighbourhood list, removes all the other nodes $v$ in that supernode for which an appropriate condition is not satisfied.
Namely, we have 
$N(u) \cup \{v\} = N(v)\cup \{u\}$ for the case of clique and
$N(u)  = N(v)$ for the case of independent set.
If the quality of the hash function is perfect, i.e. no false positives occur, then the while loop in line \ref{line:while-loop} executes only once and Algorithm~\ref{alg:filter-supernodes} is very efficient.
On the other hand, if there are false positives, then the loop will execute several times. In general, we observe that if the number of buckets is high (for our hash function we chose $2^{63}$ as number of buckets), then we only have very few false positives.   
 
Algorithm~\ref{alg:scalable-tau=1}
is the main algorithm that drives the whole process and produces the summary. 
It obtains the two hashmaps $mapC$ and $mapI$ using Algorithm~\ref{alg:candidate-supernodes} (line \ref{alg:scalable-tau=1:line3:candidatesupernodes}). 
It then removes the false positives using Algorithm~\ref{alg:filter-supernodes} (lines \ref{alg:scalable-tau=1:line:filter-cliques} and \ref{alg:scalable-tau=1:line:filter-is}). Lines \ref{alg:scalable-tau=1:line8} to \ref{alg:scalable-tau=1:line11}  handle the supernodes of size one. Finally, the superedges are built in line \ref{alg:scalable-tau=1:line:buildsuperedges}. 


\smallskip
\noindent
\textbf{Time and space complexity:} 
The work space requirement\footnote{Not considering the read-only input graph and the write-only summary graph.} of Algorithm~\ref{alg:candidate-supernodes} is only $O(V)$ due to the fact that two hashmaps $mapC$ and $mapI$ as well as list of supernodes $\mathcal{S}$ only use $O(V)$ space. The runtime is $O(E)$ as the hash function has to traverse each neighbor set of each node. 
Similarly, building superedges takes $O(V)$ space and $O(E)$ runtime. 
Algorithm~\ref{alg:filter-supernodes} takes $O(V)$ space. Its runtime, as mentioned above, depends on the quality of the hash function. 
For a perfect hash function (no false positives) this is $O(E)$. 
We observe very close to this order in practice even for simple hash functions as long as they have a large enough number of possible buckets, e.g. $2^{63}$, which is the number of possible long integers in a conventional programming language. 
To summarize, the (expected) runtime of Algorithm~\ref{alg:scalable-tau=1} is $O(E)$  and its work space requirement is $O(V)$.

\begin{algorithm}[]
    \caption{Candidate Supernodes} \label{alg:candidate-supernodes}
    \begin{algorithmic}[1]
        \State \textbf{Input:} $G=(V,E)$, $h$ \Comment{hash function to map list to number}
        \State {$mapC \gets \{ \}$ , $mapI \gets \{ \}$} \Comment{hash maps }
        \For{$v \in V$}
           \State{ $h_c \gets h((N(v) \cup \{v\})_{sorted})$} \label{line:hash-clique}
           \State{ $h_{i} \gets h(N(v)_{sorted})$} \label{line:hash-is}
           \State{ $mapC[h_c] \gets mapC[h_c] \cup \{v\} $}
           \State{ $mapI[h_{i}] \gets mapI[h_{i}] \cup \{v\} $}
        \EndFor
        \State \Return {$mapC , mapI$}
   
    \end{algorithmic}
\end{algorithm}

\begin{algorithm}[]
    \caption{Filter Supernodes} \label{alg:filter-supernodes}
    \begin{algorithmic}[1]
         \State \textbf{Input:} $map,type$ \Comment{map containing candidate supernodes}
        \State{${S}$ $\gets$ []} \Comment{list of filtered supernodes}
        \For{$ X \in values(map)$} \Comment{for each candidate supernode}
            \While{$X \neq \phi$} \label{line:while-loop}
                \State{$ u \gets$ remove-random-node($X$)}
                \If{type = clique} \label{line:filter-clique}
                    \State{$S(u) \gets \{v \in X \mid N(u) \cup \{v\} = N(v) \cup \{u\}$\} }
                \Else   $\mbox{ } S(u) \gets  \{v \in X \mid N(u) = N(v) $\} \label{line:filter-is}
                \EndIf
                \If{$\mathcal{S}(u)\neq \{u\}$} \label{alg:filter-supernodes:line9}
                    \State{$X \gets X \setminus \mathcal{S}(u)$}
                    \State{$\mathcal{S}.append(\mathcal{S}(u))$} \label{alg:filter-supernodes:line11}
                \EndIf
            \EndWhile
        \EndFor
        \State \Return $\mathcal{S}$
    \end{algorithmic}
\end{algorithm}

\begin{algorithm}[]
    \caption{Scalable algorithm for $\tau=1$}\label{alg:scalable-tau=1}
    \begin{algorithmic}[1]
        \State \textbf{Input:} $G = (V,E)$
         \State{$Status[\forall v \in V] \gets FALSE$ , $S \gets []$} \Comment{list of supernodes}
        \State{$mapC,mapI \gets$ \Call{CandidateSuperNodes}{$G$}} \label{alg:scalable-tau=1:line3:candidatesupernodes}
        \State{$C \gets \Call{FilterSuperNodes}{mapC,type=clique}$} \label{alg:scalable-tau=1:line:filter-cliques}
        \State{$I \gets \Call{FilterSuperNodes}{mapI,type=\small{independentset}}$} \label{alg:scalable-tau=1:line:filter-is}
        \State{$\mathcal{S}.append(C)$}
        \State{$\mathcal{S}.append(I)$}
        \For{$\mathcal{S}_i \in \mathcal{S}$} \label{alg:scalable-tau=1:line8}
            \For{$u \in \mathcal{S}_i$}
                {$Status[u] \gets \mbox{\sl True}$}
            \EndFor
        \EndFor
         \For{$u \in V$ \textbf{AND} $Status[u] =  \mbox{\sl False}$}
            \State{$\mathcal{S}.append(\{u\})$} \label{alg:scalable-tau=1:line11}
        \EndFor
        \State{$\Call{BuildSuperEdges}{\mathcal{S}}$} \label{alg:scalable-tau=1:line:buildsuperedges}

    \end{algorithmic}
\end{algorithm}

\subsection{How to query G-SCIS graph summaries?}

In general there are two ways to query graph summaries. 
The first is to reconstruct the original graph, then answer queries. 
Of course the reconstruction can be done incrementally and on-the-fly. 
For example, using neighborhood-queries as a primitive illustrates such a reconstruction (c.f.~\cite{5}). 
Obviously, the execution time of this approach is at least as expensive as querying the original graph. 

\looseness=-1
The second approach is to devise query answering algorithms that work directly on the summary graph and never reconstruct the original graph. This class of algorithms has the potential to produce significant gains in running time compared to executing the query on the original graph. 
Here we propose three algorithms for summaries produced by G-SCIS. 
They are for the problems of triangle enumeration, Pagerank, and shortest path queries, which form the basis for many graph-analytic tasks. 


\smallskip
\noindent
\textbf{Enumerating Triangles.} 
Triangle enumeration using G-SCIS summary is described in Algorithm~\ref{alg:TriangleEnum}. 
This algorithm can be extended to enumerating other types of graphlets, such as squares, 4-cliques, etc.
For simplicity we focus here on the case of triangles. 

\looseness=-1
As shown in Figure \ref{fig:triangle-types}, there are three different types of triangles in the summary:
\begin{description}
\item{(a)} those having all three vertices in the same supernode, 
\item{(b)} those having two vertices in one supernode and one in another, and 
\item{(c)} those having all vertices in different supernodes. 
\end{description}
The idea underlying Algorithm~\ref{alg:TriangleEnum} is to enumerate type-(a) and type-(b) triangles by iterating over the clique supernodes in $\mathcal{G}$ and generate all type-(c) triangles by considering all the supernodes (both cliques and independent sets). 

Let $X$ be a clique supernode.  
Type-(a) triangles from $X$ can be found by listing every subset of three vertices in $X$ (see lines~3 and~4).

Type-(b) triangles with two vertices in $X$ can be computed by iterating over all the super neighbors of $X$. Specifically, all such triangles can be computed by listing every subset of two vertices in $X$ combined with every subset of one vertex from a neighbor supernode $Y$ (lines \ref{alg:TriangleEnum:line5} and \ref{alg:TriangleEnum:line6}). 

Finally, any triangle enumeration algorithm can be used on the summary graph to find all the super triangles (triangles formed by three supernodes). Based on the super triangles, type-(c) triangles can be listed as follows. 
If $(X,Y,Z)$ is a super triangle, then all the corresponding type-(c) triangles can be listed by combining every choice of the first node from $X$, second node from $Y$, and third node from $Z$ (lines \ref{alg:TriangleEnum:line:superTriangles} to \ref{alg:TriangleEnum:line:enumTypeC}).

\begin{figure}
\subfigure[]{
    \resizebox{.07\textwidth}{!}{
        \begin{tikzpicture}[node distance = {0.25cm and .15cm}, every node/.style = {draw, circle}]
        \node[fill = green!40] (A3) at (5.5,0.2) {};
        \node[fill = green!40] (B3) at (5,-0.5) {};
        \node[fill = green!40] (C3) at (6,-0.5) {};
        \draw [](5.5,-0.25) circle (1cm);
        \draw[thick] (A3) --  (B3);
        \draw[thick] (A3) --  (C3);
        \draw[thick] (B3) --  (C3);
        \end{tikzpicture}
    }
}
\subfigure[]{
    \resizebox{.15\textwidth}{!}{
        \begin{tikzpicture}[node distance = {0.25cm and .15cm}, every node/.style = {draw, circle}]
        \node[fill = blue!40] (A1) at (0,0.1) {};
        \node[fill = blue!40] (B1) at (0,-0.6) {};
        \draw [](0,-0.25) circle (1 cm);
        
        \node[fill = red!40] (A2) at (2.5,-0.25) {};
        \draw [](2.3,-0.25) circle (1 cm);
        
        \draw[thick] (A1) --  (B1);
        \draw[thick] (A1) --  (A2);
        \draw[thick] (A2) --  (B1);
         
        \end{tikzpicture}
    }
 }
\subfigure[]{
    \resizebox{.12\textwidth}{!}{
        \begin{tikzpicture}[node distance = {0.25cm and .15cm}, every node/.style = {draw, circle}]
        
        \node[vertex][fill = green!40] (A) at (0,0) {};
        \node[vertex][fill = blue!40] (B) at (3,0) {};
        \node[vertex][fill = red!40] (C) at (1.5,1) {};
        \draw [](0,0) circle (0.75cm);
        \draw [](3,0) circle (0.75cm);
        \draw [](1.5,1) circle (0.75cm);
        
        \draw[thick] (A) --  (B);
        \draw[thick] (A) --  (C);
        \draw[thick] (B) --  (C);
        \end{tikzpicture}
    }
}
\caption{Three different type of triangles \label{fig:triangle-types}}
\end{figure}
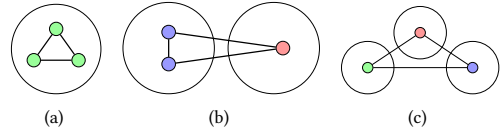

\begin{algorithm}[]
\caption{Enumerating Triangles}\label{alg:TriangleEnum}
\begin{algorithmic}[1]
    \State \textbf{Input:} $\mathcal{G}=(\mathcal{V},\mathcal{E})$ , triangle-enum \Comment{State of the art triangle enumeration algorithm}
    \For{$X \in \mathcal{V}$}
        \If{$X \in N{(X)}$} \Comment{$X$ has a superloop}
            \State{Output all \textbf{type a} triangles in $X$} \label{alg:TriangleEnum:line:enumTypeA} \Comment{$\binom{|X|}{3}$ triangles}
            \For{$Y \in \{N{(X)} \setminus X\}$} \label{alg:TriangleEnum:line5}
                \State \begin{varwidth}[t]{\linewidth}
                     Output all \textbf{type b} triangles having 2 vertices\\  \label{alg:TriangleEnum:line:enumTypeB}
                      in X and 1 vertex in Y \hspace{\algorithmicindent}\hspace{\algorithmicindent} \Comment{$\binom{|X|}{2}|Y|$ triangles}
                    \end{varwidth} \label{alg:TriangleEnum:line6}
            \EndFor
        \EndIf
        \State{super-triangles $\gets$ \textbf{triangle-enum}($\mathcal{G})$} \label{alg:TriangleEnum:line:superTriangles}
        \For{$(X,Y,Z) \in$ super-triangles}
            \State {Output all \textbf{type c} triangles in (X,Y,Z)} \label{alg:TriangleEnum:line:enumTypeC} \Comment{|X||Y||Z|  triang.} 
        \EndFor
    \EndFor
\end{algorithmic}
\end{algorithm}

\smallskip
\noindent
\textbf{Runtime Analysis.}
The running time of Algorithm~\ref{alg:TriangleEnum} is
$O(\mathcal{E}^{1.5}+\Delta)$, where $\Delta$ is the number of triangles in $G$. 
The first term, $\mathcal{E}^{1.5}$, is because of line \ref{alg:TriangleEnum:line:superTriangles}, whereas the second term, $\Delta$, is because of the enumeration we perform in lines \ref{alg:TriangleEnum:line:enumTypeA}, \ref{alg:TriangleEnum:line:enumTypeB}, and \ref{alg:TriangleEnum:line:enumTypeC}.
The running time of the rest of the steps of the algorithm add up to $O(\mathcal{V} + \mathcal{E})$ time which is absorbed by $O(\mathcal{E}^{1.5}+\Delta)$. Therefore, the latter expression gives us the running time of Algorithm~\ref{alg:TriangleEnum}. 

However, if the task is just counting the number of triangles, then lines \ref{alg:TriangleEnum:line:enumTypeA}, \ref{alg:TriangleEnum:line:enumTypeB}, and \ref{alg:TriangleEnum:line:enumTypeC} are $O(1)$ operations (calculating the triangle numbers given in code-comments) and the running time is $O(\mathcal{E}^{1.5})$.
As  $|\mathcal{E}|$ is significantly smaller than $|E|$ and the fact that enumerating triangles directly on $G$ requires $O(|E|^{1.5})$ operations,  performing triangle enumeration directly on $\mathcal{G}$ makes it faster. 
In experiments, we compare the running time of enumerating triangles using G-SCIS summary versus using the original graph. 
We employ a state-of-the-art algorithm for triangle enumeration \cite{santoso2019triad} and 
validate our claim.


\medskip
\looseness=-1
\noindent
\textbf{Computing Pagerank.} 
Another interesting application of the G-SCIS summary is that it can be used to find the Pagerank scores of all nodes in $G$ without reconstructing $G$. 
Before describing our approach, we give the definition of Pagerank and state a nice property of G-SCIS summary in Theorem \ref{thm:same-Pagerank} below.

Let $P^{i}(u)$ denote the Pagerank value of any node $u$ after $i$-th iteration of the Pagerank algorithm \cite{Page1999ThePC}. For any undirected graph $G=(V,E)$, all the nodes are initialized with the same Pagerank value i.e. $\forall_{u \in V}P^0(u)$ = 1. In iteration $i$, it is updated as follows:
\begin{equation} \label{eq:pr-definition}
    P^{i}(u) \gets \sum_{w \in N(u)}{\frac{P^{i-1}(w)}{|N(w)|}}  
\end{equation}

In Equation \ref{eq:pr-definition}, we ignore damping factor for simplicity  but it can be easily incorporated without impacting our results.

\begin{theorem} \label{thm:same-Pagerank}
For any supernode $S \in \mathcal{V}$, all the nodes inside $S$ must have the same Pagerank value.
\end{theorem}
\begin{proof}
As the supernodes in the lossless summary $\mathcal{G}$ either represent an independent set or a clique, we show that in both cases this property holds true. 
\begin{enumerate}
    \item For any two nodes $u,v$ in an independent set supernode, $N(u)$ is exactly the same as $N(v)$ and according to Equation~\ref{eq:pr-definition}, $P^{i}(u)$ = $P^{i}(v)$.
    \item For any two nodes $u,v$ in a clique supernode, Equation \ref{eq:pr-definition} can be rewritten as follows: 
    \begin{equation} \label{eq: pr-clique-definition}
        \begin{split}
            P^{i}(u) \gets \frac{P^{i-1}(v)}{|N(v)|} + \sum_{w \in N(u) \setminus \{v\}}{\frac{P^{i-1}(w)}{|N(w)|}} \\
            P^{i}(v) \gets \frac{P^{i-1}(u)}{|N(u)|} + \sum_{w \in N(v) \setminus \{u\}}{\frac{P^{i-1}(w)}{|N(w)|}}  
        \end{split}
    \end{equation}
    
    From the properties of clique supernode, $N(u) \setminus \{v\}$ = $N(v) \setminus \{u\}$ and  $|N(u)| = |N(v)|$. Thus, it can be seen from Equation \ref{eq: pr-clique-definition} that  $P^{i}(u) = P^{i}(v)$ if and only if $P^{i-1}(u) = P^{i-1}(v)$.
    
    $P^{0}(v) = P^{0}(u)$ and for any iteration $k$, if we assume $P^{k}(u) = P^{k}(v)$, Equation \ref{eq: pr-clique-definition} implies that $P^{k+1}(u) = P^{k+1}(v)$. Thus, by induction,  $P^{i}(u) = P^{i}(v)$ for all $i$.
\end{enumerate}
\end{proof}
\looseness=-1
 To calculate the exact Pagerank scores of the nodes in $G$ using its summary $\mathcal{G}$, we propose Algorithm \ref{alg:Pagerank-gsic}, an adaptation of the Pagerank algorithm, that runs directly on $\mathcal{G}$ and prove its correctness in Theorem \ref{thm:Pagerank-same-result}. 
Algorithm \ref{alg:Pagerank-gsic} maintains the invariant that the Pagerank of a supernode after iteration $i$ is the sum of the Pagerank of its nodes after iteration $i$ of the Pagerank algorithm. It initializes the Pagerank of a supernode to be its size (line \ref{alg:Pagerank-gsic:line:initialization}). It computes the number of neighbours of a node inside a supernode $X$ (lines \ref{alg:Pagerank-gsic:line4} and \ref{alg:Pagerank-gsic:line5}). Using this, it updates the Pagerank of supernode $X$ in iteration $i$ (line \ref{alg:Pagerank-gsic:line8} to \ref{alg:Pagerank-gsic:line11}). Finally, it computes the Pagerank of each node of $G$ from the Pagerank of its supernode in $\mathcal{G}$ (line \ref{alg:Pagerank-gsic:line14}).
\begin{algorithm}[]
\caption{Pagerank using G-SCIS summary}\label{alg:Pagerank-gsic}
\begin{algorithmic}[1]
    \State \textbf{Input: } $\mathcal{G}=(\mathcal{V},\mathcal{E})$ 
    \State \textbf{Initialization:} $\forall{X \in \mathcal{V}}, {P^0(X)=|X|} , i \gets 1$ \label{alg:Pagerank-gsic:line:initialization}
    \For{$X \in \mathcal{V}$}
        \If{$X \notin N(X)$} 
            \state $W(X) \gets \sum_{Y \in N^-(X)}{|Y|}$  \Comment {X is IS} \label{alg:Pagerank-gsic:line4}
        \Else
            \state{$W(X) \gets \sum_{Y \in N^-(X)}{|Y|} + (|X|-1)$} \Comment{X is clique} \label{alg:Pagerank-gsic:line5}
        \EndIf
    \EndFor
    \While{$P^i \neq P^{i-1} $} \Comment{until convergence}
        \For{$X \in \mathcal{V}$}
            \If{$X \notin N(X)$} \Comment{X is IS} \label{alg:Pagerank-gsic:line8}
            \State{$P^{i}(X) \gets \sum_{Y \in N^-(X)}\frac{|X| \cdot P^{i-1}(Y)}{W(Y)}$}
            \Else \Comment{X is clique}
                \State{$P^{i}(X) \gets \sum_{Y \in N^-(X)}\frac{|X| \cdot P^{i-1}(Y)}{W(Y)}+ \frac{(|X|-1) \cdot P^{i-1}(X)}{W(X)}$}
                \label{alg:Pagerank-gsic:line11}
            \EndIf
        \EndFor
    \EndWhile
    
    \For{$X \in \mathcal{V}$}
        \For{$u \in nodes(X)$}
            \State{$P(u) \gets \frac{P^i(X)}{|X|}$} \label{alg:Pagerank-gsic:line14}
        \EndFor
    \EndFor
\State return $P$    
\end{algorithmic}
\end{algorithm}





\begin{theorem} \label{thm:Pagerank-same-result}
Algorithm \ref{alg:Pagerank-gsic} outputs exactly the same Pagerank score for each node $v$ in $G$ as the Pagerank algorithm.
\end{theorem}
\begin{proof}
Let $N^-(X)$ denote the set $N(X) \setminus \{X\}$.
Replacing the role of $G$ with $\mathcal{G}$, Equation \ref{eq:pr-definition} can be rewritten as follows:
\begin{equation} \label{eq:pr-definition-modified}
P^i(u) = \begin{cases}
    \sum\limits_{Y \in N^-(S(u))}{\sum\limits_{ w \in Y}{\frac{P^{i-1}(w)}{|N(w)|}}}
    &\text{$(S(u)$ is IS)} \\
    \sum\limits_{Y \in N^-(S(u))  }{\sum\limits_{ w \in Y}{\frac{P^{i-1}(w)}{|N(w)|}}} + \sum\limits_{\substack{w \in S(u) \\ w \neq u }}{\frac{P^{i-1}(w)}{|N(w)|}}
    &\text{ $(S(u)$ is clique)} 
    \end{cases}
\end{equation}
From Theorem \ref{thm:same-Pagerank}, all the nodes in a supernode S have same Pagerank scores. Hence, $P(X) = |X|\cdot P(w)$ for any node $w \in X$.
Also, observe that all the nodes in a supernode have same number of neighbors in $G$. i.e $\forall_{u,v \in X}{|N(u)|=|N(v)|}$. Let $W(X)$ represent the number of neighbors of any node inside a supernode $X$. Then, 

\begin{equation}\label{eq:w(s)-def}
 W(X)= \begin{cases}
    \sum\limits_{Y \in N^-(X)}{|Y|}  
    &\text{($X$ is independent set)} \\
     \sum\limits_{Y \in N^-(X)}{|Y|} + (|X|-1) 
    &\text{($X$ is clique)}
    \end{cases}
\end{equation}

Using Equation \ref{eq:w(s)-def}, Equation \ref{eq:pr-definition-modified} can be finally rewritten as
\begin{equation} \label{eq:pr-definition-modified-modified}
     P^i(u)= \begin{cases}
    \sum\limits_{Y \in N^-(S(u)) }{\frac{P^{i-1}(Y)}{W(Y)}}
    &\text{($S(u)$ is IS)} \\
    \sum\limits_{Y \in N^-(S(u)) }{\frac{P^{i-1}(Y)}{W(Y)}}  +\frac{(|S(u)|-1) \cdot P^{i-1}{(S(u))}}{|S(u)| \cdot W(S(u))}
    &\text{($S(u)$ is clique)}
    \end{cases}
\end{equation}
The output of Algorithm \ref{alg:Pagerank-gsic} (line \ref{alg:Pagerank-gsic:line14}) is exactly the same as Equation~\ref{eq:pr-definition-modified-modified}. The correctness follows.
\end{proof}

\smallskip
\noindent
\looseness=-1
\textbf{Runtime Analysis.} 
The time complexity of the Pagerank algorithm is $O(T(|V|+|E|))$ where $T$ is the number of iterations required for convergence. 
If we run Algorithm \ref{alg:Pagerank-gsic} on $\mathcal{G = (V,E)}$, the running time is $O(|V|+T(\mathcal{|V|+|E|}))$ as the number of iterations are the same in both the algorithms. 
Thus, the running time is reduced from $O(T(|V|+|E|))$ to $O(|V|+T(\mathcal{|V|+|E|}))$. 

\smallskip
\noindent
\textbf{Computing Shortest Paths.} 
We observe that G-SCIS summary, $\mathcal{G}$, can be used to compute lengths of shortest paths between any two nodes $u,v \in G$ (as an unweighted graph). To find the shortest paths in $G$, BFS can be executed directly on  $\mathcal{G}$ reducing runtime from $O(|E|+|V|)$ to $O(\mathcal{|E|+|V|})$. 
We present the following two theorems. 

\begin{theorem}
Given nodes $u,v$ such that $S(u) = S(v)$, the following hold.
\begin{enumerate}
\item If $S(u)$ is a clique, the shortest path length between $u$ and $v$ in $G$, $d(u,v)$, is $1$. 
\item If $S(u)$ is an independent set and $|N(S(u))|>0$, then $d(u,v) = 2$. Otherwise, $d(u,v) = \infty$.
\end{enumerate}
\end{theorem}

\begin{theorem}
Given nodes $u,v$ such that $S(u) \neq S(v)$, $d(u,v)$ is equal to the length of shortest path between $S(u)$ and $S(v)$ in $\mathcal{G}$. 
\end{theorem}
\begin{proof}
We observe that any two nodes on the shortest path between $u$ and $v$ in $G$ cannot be in the same supernode of $\mathcal{G}$. Otherwise, as nodes in a supernode have the same connectivity, it can be seen that that path is not the shortest. Thus, the shortest path between $u$ and $v$ can pass only once through each supernode on the shortest path between $S(u)$ and $S(v)$ and hence the lengths of both paths are the same.
\end{proof}

Based on the above theorems we present Algorithm~\ref{alg:shortestpath-gsic} for computing $d(u,v)$ given two nodes $u,v\in V$ using a G-SCIS summary.

\begin{algorithm}[]
\caption{Shortest Paths using G-SCIS summary}\label{alg:shortestpath-gsic}
\begin{algorithmic}[1]
    \State \textbf{Input: } $\mathcal{G}=(\mathcal{V},\mathcal{E}), u,v\in V$ 
        \If{$S(u) = S(v)$} 
           \If{$S(u)$ is a clique}
                \state $d(u,v) = 1$
            \Else
                \If{$N(S(u) > 0$}
                    \state $d(u,v) = 2$
                \Else \state $d(u,v) = \infty$
                \EndIf
            \EndIf
        \Else
            \state{$d(u,v) = d(S(u), S(v))$}
        \EndIf
\State return $d(u,v)$
\end{algorithmic}
\end{algorithm}

\section{Scalable Lossy Algorithm}
Kumar et al. \cite{1} proposed UDS, a lossy algorithm for Utility Driven graph Summarization. 
We introduce UDS briefly because of its relevance to
our work, and highlight some of its limitations. 

\looseness=-1
UDS is a greedy iterative algorithm that starts with the original graph $G=(V,E)$ and iteratively merges nodes until the utility of the graph drops below a user-specified threshold $\tau < 1$. 
Intuitively, it is desirable that any two nodes in the same supernode have similar neighborhoods. 
A good starting point advocated in \cite{1} for the greedy algorithm is to look at the two-hop away nodes, as they have at least one neighbor in common, and merge them together in a supernode. 
To decide the order of the merge operations, UDS considers the set of all two-hop away nodes as the candidate pairs. Call this set $F$.
The algorithm starts merging from the less central candidate pairs in $F$ because they result in less damage to the utility. 
Towards this, UDS uses a centrality score for each node in the graph to assign a weight to each candidate pair $\langle u,v \rangle$, e.g. $C_u + C_v$ and sorts them in ascending order. 
UDS iterates over the sorted candidate pairs and in each iteration performs the following steps.
%
\begin{enumerate}
     \item Pick the next pair of candidate nodes $\langle u,v\rangle$ from $F$, find their corresponding supernodes  $S(u),S(v)$, and merge them into a new supernode $S$, if $S(u)\neq S(v)$.
    \item Update the neighbors of $S$ based on the neighbors of $S(u),S(v)$. In particular, add an edge from $S$ to another supernode if the loss in utility is less than the loss if not added. 
    
    \item (Re)compute the utility of the summary built so far and stop if the threshold is reached. 
\end{enumerate}


UDS needs $O(|F|\lg{(|F|)})$ time and $O(|F|)$ space to compute and sort $F$. Merging two supernodes takes $O(V)$ time and there can be $O(V)$ such merges. Therefore, merge steps together require $O(|V|^2)$ time and $(|E|)$ space. Thus, the time complexity of UDS is  $O(|V|^2 + |F| \cdot \lg(|F|))$ and its space complexity is $O(|F|)$.
This complexity, particularly running time, makes UDS impractical for large graphs.

    




Now, we introduce our new algorithm that overcomes the limitations of UDS. It makes use of two techniques, constructing a minimum spanning tree (MST) of the two-hop graph and performing a binary search on the list of MST edges.
We call our approach T-BUDS.\footnote{msT-Binary search based Utility Driven graph Summarization (T-BUDS).} %
As shown in the experiments section, T-BUDS outperforms UDS by two orders of magnitude on several moderate datasets, while on 
bigger ones, it is only T-BUDS that can complete the computation in reasonable 
time (UDS could not). 



\subsection{Generating Candidate Pairs using MST \label{subsect:2-hopCandidate}} 
Recall that UDS considers the set of all two-hop away nodes as candidate pairs starting from the less central to more central pairs. However, not every candidate pair will cause a merge. This is because the nodes in the pair can be already in a supernode together due to previous merges. Therefore, there are many useless pairs, which we eliminate with our MST technique below. 

We denote the two hop graph by $G_{2-hop} =(V,F)$. 
That is, 
$F = \{(a,c)|(a,b)\in E \mbox{ and } (b,c) \in E\}$.
We do not construct it explicitly as UDS does. 
We propose a method to reduce the number of candidate  pairs from $O(|F|)$ to $O(|V|)$ by creating an MST of $G_{2-hop}$.
In Theorem~\ref{thm:mst-sufficiency}, we prove that using  the sorted edge list of MST of $G_{2-hop}$ will produce exactly the same summary as using the sorted edge list of $G_{2-hop}$. 

\looseness=-1
Let us denote by $L$ the centrality weight-based sorted version of $F$.
Also, we denote by $H$ the sorted list of edges of an MST for $G_{2-hop}$.  
We now present a sufficiency theorem, which says that using $H$ instead of $L$ as the list of candidates is sufficient. The idea of the proof is that the candidate pairs leading to a merge when $L$ is used, in fact, exactly correspond to the edges of an MST. 

\begin{theorem}[\textbf{MST Sufficiency Theorem}]\label{thm:mst-sufficiency}
For utility threshold $\tau$,
using $H$ as the list of candidate pairs will produce the same graph summary as using $L$.\footnote{
There can be different sorted versions of $L$ due to possible ties (albeit unlikely as weights are real numbers). 
What this theorem shows is that the summary constructed based on MST is the same as the summary constructed using {\em some} sorted version of~$L$.}
\end{theorem}

\begin{proof}
  Initially $\mathcal{G}$ is same as $G$ and let us assume that at iteration $i$  a new pair  $\langle u ,v \rangle \gets L[i]$ is chosen and  ${S(u)}^{i-1}$ and ${S(v)}^{i-1}$ are their corresponding supernodes. If ${S(u)}^{i-1}\neq {S(v)}^{i-1}$ then they should be merged together into a new supernode. The following two claims need to be proven to ensure the sufficiency of $H$ as a candidate set.
  

\begin{enumerate}
     \item If two pairs $\langle u_1,v_1 \rangle$ and $\langle u_2,v_2 \rangle$ are in $H$
     such that $\langle u_1,v_1 \rangle$ appears before $\langle u_2,v_2 \rangle$ in $H$ then $\langle u_1,v_1 \rangle$ appears before $\langle u_2,v_2 \rangle$ in $L$.
     
     \item If $u$ and $v$ are not inside a same supernode, that is ${S(u)}^{i-1} \neq {S(v)}^{i-1}$, then $\langle u ,v \rangle$ must be in $H$.
     
 \end{enumerate}
 
 Proof of (1): As both $H$ and $L$ are sorted based on the weights of the edges, the order in which $\langle u_1,v_1 \rangle$ and $\langle u_2,v_2 \rangle$ appear in $H$ will be the same as their order in $L$.

 Proof of (2): $S(u)^{i-1} \neq S(v)^{i-1}$ implies that there does not exist any other pair  $\langle u',v' \rangle \gets L[j]$ for any $j<i$ such that $u' \in S(u)^{i-1}$ and $v' \in S(v)^{i-1}$. 
Otherwise, $S(u')^{j}$ would have been merged with $ S(v')^{j}$ in the $j$-th iteration. Thus, $u'$ and $v'$ would belong to the same supernode and $S(u)^{i-1}$ should be same as $S(v)^{i-1}$.
Hence, $\langle u ,v \rangle$ is the smallest weight edge in $G_{2-hop}$ connecting $S(u)^{i-1}$ and $S(v)^{i-1}$. 
We want to show now that $\langle u ,v \rangle \in H$ i.e. part of the MST.
To show this, we claim that, in fact, $\langle u ,v \rangle$ is the smallest weight edge in $G_{2-hop}$ connecting $S(u)^{i-1}$ and $V \setminus S(u)^{i-1}$. 
Suppose not. Let us consider the edges between $S(u)^{i-1}$ and $V \setminus S(u)^{i-1}$. Recall that a cut in a connected graph is a minimal set of edges whose removal disconnects the graph. Therefore, the edges between $S(u)^{i-1}$ and $ V \setminus S(u)^{i-1}$ form a cut in $G_{2-hop}$. A well known property called {\em cut property} of MST  states that the minimum weight edge of any cut belongs to the MST \cite{kruskal1956shortest}. Now let, if possible, a different edge, $\langle u'' ,v'' \rangle$ in $G_{2-hop}$ be the edge with the smallest weight connecting $S(u)^{i-1}$ and $V \setminus S(u)^{i-1}$.  Then by the cut property, $\langle u'',v'' \rangle$ belongs to $H$ and would have been considered as a candidate pair for merge in an earlier iteration. In that case, $u''$ and $v''$ will belong to the same supernode which is a contradiction. 
\end{proof}

\subsection{Scalable Binary Search based Algorithm}

Based on Theorem~\ref{thm:mst-sufficiency}, we can use $H$ instead of $L$ for the list of candidate pairs. Furthermore, we show in following theorem that the utility is non-increasing as we merge candidate pairs of $H$ in order. 
\begin{theorem}[\textbf{Non-increasing utility theorem}]\label{thm:utility}
Let $\mathcal{G}^0=G$ and $\mathcal{G}^{t}$ be the summary graph obtained by processing $H$ in order from index 1 to $t$ where $1 \leq t \leq |H|$. 
Then 
$u(\mathcal{G}^{t-1})$ $\geq$ $u(\mathcal{G}^{t})$. 

\end{theorem}

\begin{proof} 
Suppose at iteration $t$, we take a  pair $\langle u,v\rangle \gets H[t]$  and two supernodes ${S(u)}$ and ${S(v)}$ be merged together to create a new supernode $S$. After the merge, all the superedges between $S$ and $S(w) \in N({S})$ should be updated where $N({S})$ is the set of supernodes such that $\exists (u,v) \in E$ for any $u \in S$ and $v \in S(w)$. 

\looseness=-1
Let $\Sedge(S_i,S_j)$ be the cost of  adding a superedge between $S_i$ and $S_j$. As some spurious edges are introduced by adding a superedge, the cost includes the cost of all those spurious edges. Similarly, 
let $\nSedge(S_i,S_j)$ be the cost of not adding a superedge between $S_i$ and $S_j$. As some actual edges are missed by not adding a superedge, the cost includes the cost of all those actual edges.
We have
\begin{equation}
\Sedge(S_i,S_j) = \sum_{\substack{(u,v)\notin E \\ u \in S_i, v \in S_j} } C_s(u,v)
\end{equation}

\begin{equation}
\nSedge(S_i,S_j) = \sum_{\substack{(u,v)\in E \\ u \in S_i, v \in S_j} } C(u,v)
\end{equation}

Note that, at iteration $t$, when two supernodes $S(u)$ and $S(v)$ are merged together into $S$, then the number of spurious edges introduced on adding a superedge between $S$ and any neighbor $S(w)$ is exactly equal to the the sum of spurious edges introduced on connecting $S(u)$, $S(w)$ and $S(v)$, $S(w)$. When $S(w) \neq S$,
the cost of adding a superedge between $S$ and $S(w)$, $\mbox{Sedge}(S,S(w))$, and the cost of not adding a superedge between $S$ and $S(w)$, \\ $\mbox{nSedge}(S,S(w))$, can be calculated as follows:
\begin{equation}
  \Sedge(S,S(w)) = \Sedge(S(u),S(w)) + \Sedge(S(v),S(w))  \label{eq:Sedge}
\end{equation}
\begin{equation}
  \nSedge(S,S(w)) = \nSedge(S(u),S(w)) + \nSedge(S(v),S(w))  \label{eq:nSedge}
\end{equation}

 Let us represent $\mbox{loss}(S(w))$ as the smaller of the cost of adding or not adding a superedge between new supernode $S$ and any other candidate super neighbor $S(w)$. Formally, it is defined by:

\begin{equation} \label{eq:loss-s(w)}
    \begin{split}
        \mbox{loss}(S(&w)) = min(\Sedge(S,S(w)),\nSedge(S,S(w)))  \\ 
          -&min(\Sedge(S(u),S(w)) ,  \nSedge(S(u),S(w))) \\    -&min(\Sedge(S(v),S(w)),  \nSedge(S(v),S(w))
    \end{split}
\end{equation}

\looseness=-1
Let us denote $\Sedge(S(u),S(w))$ as $a$, $\nSedge(S(u),S(w))$ as $b$, $\Sedge(S(v),S(w))$ as $c$, and $\nSedge(S(v),S(w))$ as $d$. Then, Equation \ref{eq:loss-s(w)} is of the form $\min(a+c,b+d)- (\min(a,b) + \min(c,d))$. 
As $\min(a+c,b+d) \geq \min(a,b) + \min(c,d)$, we have $\mbox{loss}(S(w)) \geq 0$. 
So,
$u(\mathcal{G}^{t-1}) - u(\mathcal{G}^{t}) =\sum_{S(w) \in {N(S(u))\cup N(S(v))}} \mbox{loss}(S(w)) \geq 0.$

We can follow a similar strategy for the case of a superloop in which $S(w) = S$ and show that 
$u(\mathcal{G}^{t-1}) \geq u(\mathcal{G}^{t})$.
\end{proof}

Theorems \ref{thm:mst-sufficiency} and \ref{thm:utility} form the basis of our new approach T-BUDS that uses binary search over the sorted list of MST edges, $H$, in order to find the largest index $t$ for which $u(\mathcal{G}^t) \geq \tau$
(see Algorithm~\ref{alg:BinaryS}). 
This requires computing $H$ (done using Algorithm~\ref{alg:2-hopMST}) followed by $\lg(|H|)$ computations of utility. The latter is done using Algorithm~\ref{alg:computeUtility}.

Given graph $G = (V,E)$ and centrality scores for each node $C[u \in V]$, T-BUDS first creates the sorted candidate pairs $H$ by calling the Two-hop MST function (Algorithm~\ref{alg:2-hopMST}). 
This function follows the structure of Prim's algorithm~\cite{prim1957shortest} for computing MST. 
However, we do not want to build the $G_{2-hop}$ graph explicitly. 
As such, we start with an arbitrary 
node $s$ and insert it into a priority queue $Q$ with a key value of 0. All other nodes are initialized with a key value of $\infty$. 
For any given node $v$ with minimum key value deleted from $Q$, $v$ is included in the MST, and the key values of its two-hop away neighbours are updated, when needed. 

After creating the two-hop MST and sorting its edges, T-BUDS uses a binary search approach and iteratively performs merge operations from the first pair until the middle pair in $H$ (Algorithm~\ref{alg:BinaryS}). 
In each iteration, we pick a pair of nodes $u,v$ from $H$, find their supernodes $S(u)$ and $S(v)$ and merge them into a new supernode~$S$. 
This process continues until the algorithm reaches the middle point. $\mathcal{G}$ is the resulting summary after these operations and we compute its utility in line~11. 
If this utility $\geq \tau$, then we search for the index $t$ in the second half, otherwise, we search for the index $t$ in the first half. 
The algorithm finds the best summary in $\lg{|H|}$ iterations and $|H|$, being the number of edges in the MST of $G_{2-hop}$, is just $O(V)$.

\begin{algorithm}[]
    \caption{T-BUDS}\label{alg:BinaryS}
    \begin{algorithmic}[1]
    \State \textbf{Input:} $G = (V,E), C, \tau $ \State{$H \gets \mbox{\Call{TwoHopMST}{G,C}}$}
    \State $low \gets 0,high \gets |H|-1$
    \While{$low \leq high$}
        \State $mid \gets  \frac{low+high}{2}$
        \State {$ \mathcal{V} \gets V$, $i \gets 0$}
    \For{$i \leq mid$}
        \State{$\langle u,v \rangle \gets H[i]$, $i \gets i+1$}
        \State{$S \gets$ \Call{Merge}{$S(u),S(v)$}} \label{alg:BinaryS:mergeline}
        \State{$\mathcal{V} \gets \left(\mathcal{V} \setminus \{S(u), S(v)\}\right) \cup S$}
    \EndFor
    \State{$u(\mathcal{G})\gets$ \Call{ComputeUtility}{ $\mathcal{V}$}}
    \If{$ u(\mathcal{G}) \geq \tau$}{  $high=mid-1$}
    \Else {  $low=mid+1$}
    \EndIf
    \EndWhile
    \State{$\Call{BuildSuperEdges}{\mathcal{V}}$}
\end{algorithmic}
\end{algorithm}
 Algorithm \ref{alg:computeUtility} is used to compute the utility for a specific summary $\mathcal{G} = (\mathcal{V},\mathcal{E})$.
 The algorithm iterates over all supernodes one at a time and for a given supernode $S_i$, it creates two maps ($count$ and $sum$) to hold the details for the superedges connected to $S_i$.  $count[S_j]$ stores the number of actual edges between supernodes $S_i$ and $S_j$. Similarly, $sum[S_j]$ contains the sum of the weights for all the edges between $S_i$ and $S_j$. Line \ref{alg:computeUtility:line5} to line \ref{alg:computeUtility:line13} initialize these two structures. $\Sedge(S_i,S_j)$ (the cost of drawing a super edge between $S_i$ and $S_j$) and $\nSedge(S_i,S_j)$ (the cost of not drawing a super edge between $S_i$ and $S_j$) can be estimated using $count$ and  $sum$.
As $\nSedge(S_i,S_j)$ is the sum of weights of edges in $G$ between nodes in $S_i$ and $S_j$, it is exactly equal to $sum[S_j]$ (line \ref{alg:computeUtility:line15}).  
If $S_i \neq S_j$, the number of spurious edges is equal to $|S_i||S_j|-count[S_j]$ and since each spurious edge has cost $\frac{1}{\binom{|V|}{2} - |E|}$, \\ $\Sedge(S_i,S_j) = \frac{|S_i||S_j|-count{[S_j]}}{\binom{|V|}{2} - |E|}$ (line \ref{alg:computeUtility:line16}).
Similarly, if $S_i = S_j$, the number of spurious edges is $\binom{|S_i|}{2}-count{[S_j]}$ and $\Sedge(S_i,S_j) = \frac{\binom{|S_i|}{2}-count{[S_j]}}{\binom{|V|}{2} - |E|}$ (line \ref{alg:computeUtility:line17}).
Finally the utility loss can be estimated as $\min(\Sedge(S_i,S_j)$ , $\nSedge(S_i,S_j))$ and the utility is decremented by the loss. Algorithm \ref{alg:computeUtility} returns the final utility for $\mathcal{G}$ which is used by Algorithm \ref{alg:BinaryS} for making decisions.

\smallskip
\noindent
{\bf Building Superedges.}
Once the appropriate supernodes have been identified, a superedge is added between two supernodes $S_i$ and $S_j$ if and only if {$\Sedge(S_i,S_j) \leq \nSedge(S_i,S_j)$}. This task can be completed in $O(|E)$ time: Line \ref{alg:computeUtility:line19} of Algorithm \ref{alg:computeUtility} can be replaced by the task of adding superedge between $S_i$ and $S_j$.

\begin{algorithm}[]
    \caption{Two-hop MST}\label{alg:2-hopMST}
    \begin{algorithmic}[1]
    \State \textbf{Input:} $G = (V,E), C $ \Comment{$C$ is centrality scores array for nodes}
    \State{ $key[s] \gets 0$,
            $parent[s] \gets \mbox{\sl Null}$,
            $Q.insert(s,key[s])$}
     \For{($v \in V\setminus\{s\}$)}
     \State{$key[v] \gets \infty$,
            $parent[v] \gets \mbox{\sl Null}$, 
            $Q.insert(v,key[v])$}
     \EndFor
     \While{$!isEmpty(Q)$}
     \State{$(v,\_) = Q.delMin()$}
     \For{($w \in N(N(v)) \mid w \in Q \And w \neq v$)}
     \If{$key[w] > C[v]+C[w]$}
     \State{$Q.setKey(w, C[v] + C[w])$}
     \state{$parent[w] \gets v$}
     \EndIf
     \EndFor
     \EndWhile
     \State $H \gets \{(v,parent[v]) : v \in V\setminus\{s\} \}$
     \State \Return{sorted $H$ based on $C$}
\end{algorithmic}
\end{algorithm}
\begin{algorithm}[]
\caption{Compute Utility}\label{alg:computeUtility}
\begin{algorithmic}[1]
 \State \textbf{Input:} $G=(V,E),utility \gets 1, \mathcal{V}$ \Comment{set of supernodes}
    \For{$S_i \in \mathcal{V}$} \Comment{for each supernode}
        \State{$count \gets \{\}$}, {$sum \gets \{\}$} 
        \For{$u \in S_i$} \label{alg:computeUtility:line5}
            \For{$v \in N(u)$} 
                \State{$S_j \gets S(v)$}  \label{alg:computeUtility:findLine}
                \If{$(S_i \neq S_j ) \lor (S_i = S_j \land i < j)$}
                    \If{$count[S_j] \geq 1$}
                        \State{$count{[S_j]} \gets count{[S_j]} + 1$}
                        \State{$sum{[S_j]} \gets sum{[S_j]} + C(u,v)$}
                    \Else
                        \State{$count{[S_j]} \gets 1$} 
                        \State{$sum{[S_j]} \gets C(u,v)$} \label{alg:computeUtility:line13}
                    \EndIf
                \EndIf
            \EndFor
        \EndFor
        \For{$S_j \in count.keys \land i \leq j $} \label{alg:computeUtility:line14}
            \State{$\nSedge(S_i,S_j) \gets sum{[S_j]}$}  \label{alg:computeUtility:line15}           
            \If{$S_i \neq S_j$} {$\Sedge(S_i,S_j) \gets \frac{|S_i||S_j|-count{[S_j]}}{\binom{|V|}{2} - |E|}$} \label{alg:computeUtility:line16}
            
            \Else {  $\Sedge(S_i,S_j) \gets \frac{\binom{|S_i|}{2}-count{[S_j]}}{\binom{|V|}{2} - |E|}$} \label{alg:computeUtility:line17}  
            \EndIf
            \If{$\Sedge(S_i,S_j) \leq \nSedge(S_i,S_j)$}
                \State{$utility \gets utility - \Sedge(S_i,S_j)$} \label{alg:computeUtility:line19}
            \Else {  $utility \gets utility - \nSedge(S_i,S_j)$}
            \EndIf
        \EndFor
    \EndFor
    \State{\textbf{return} $utility$}
\end{algorithmic}
\end{algorithm}

\smallskip
\noindent
{\bf Data structures.}
We used the union-find algorithm \cite{hopcraft1973set} for representing our supernodes. The union operation was used to implement the merge operation in line~\ref{alg:BinaryS:mergeline} of Algorithm~\ref{alg:BinaryS} and the find operation was used to find the corresponding supernode for a specific node in  line~\ref{alg:BinaryS:mergeline} of Algorithm \ref{alg:BinaryS} and line \ref{alg:computeUtility:findLine} of Algorithm \ref{alg:computeUtility}. 
Using path compression with the union-find algorithm allows reducing the complexity of the union and find operations to $O(\lg^{\star}{|V|)}$ (iterated logarithm of $|V|$). 
As $\lg^{\star}{|V|}$ is about 5 when $|V|$ is even more than a billion, we treat it as a constant in our calculations. The union-find algorithm only needs two arrays of size $|V|$ and thus the working memory requirement is $O(|V|)$.  

\subsection{Complexity analysis}










Let us begin by analysing the time complexity of Algorithm~\ref{alg:2-hopMST}. 
As its structure follows that of Prim's algorithm \cite{prim1957shortest}, 
it requires $O(|F| \cdot \lg|V|)$ steps to compute MST. 
As the number of edges in $H$ is $O(|V|)$, sorting it takes $O(|V|\lg|V|)$ time. 
Thus, the total time complexity of Algorithm~\ref{alg:2-hopMST} is $O((|F| + |V|)\cdot \lg|V|)$.
The total space required by Algorithm~\ref{alg:2-hopMST} is $O(|V|)$ as it stores the priority queue $Q$ and arrays $key$, $parents$, and $H$ all of size $O(|V|)$.

Now let us analyse the time complexity of Algorithm~\ref{alg:computeUtility}. To compute the utility of $\mathcal{G}$, the algorithm iterates over all the edges in $G$, each edge exactly once, to identify pairs of supernodes $(S_i, S_j)$ that have at least one edge of $G$ between them. This step, that includes the computation of {\em count} and {\em sum} for each supernode, takes $O(E)$ time.  Once this step is completed, it takes $O(1)$ time to compute the $\Sedge$ and  $\nSedge$ cost for a pair $(S_i, S_j)$. Therefore, the time complexity of Algorithm~\ref{alg:computeUtility} is $O(|E|)$. 
It requires $O(|V|)$ space to store the count and sum arrays. 

Finally, let us analyse the time and space complexity of Algorithm~\ref{alg:BinaryS}. As discussed in Section 5.2, Algorithm \ref{alg:BinaryS} will perform $\lg{|H|}$ iterations. For each iteration, merging supernodes in Algorithm \ref{alg:BinaryS} requires $O(|H|)$ operations 
and the utility estimation using Algorithm \ref{alg:computeUtility}  requires $O(|E|)$ time. Thus the time complexity for each iteration is $O((|E|+|V|)$ and time for a total of $\lg{|H|}$ iterations is $O((|E|+|V|) \cdot \lg(|V|))$.
The space requirement inside Algorithm \ref{alg:BinaryS} is storing $H$ and $\mathcal{V}$, which is $O(|V|)$. Thus, the space requirement of Algorithm \ref{alg:BinaryS} is $O(|V|)$.
Summarizing all the above, we have

\begin{theorem}
The time complexity of T-BUDS is $O((|F| + |V|)\cdot \lg|V|)$.
The space complexity of T-BUDS is $O(|V|)$.
\end{theorem}

\section{Experiments \label{sec:Experiments}}
The experimental evaluation is divided into the following four  parts:
\begin{enumerate}
    \item Performance analysis of G-SCIS versus SWeG \cite{5} (state-of-the-art in lossless graph summarization) in terms of running time and node reduction.
    \item Performance analysis of G-SCIS on triangle enumeration and Pagerank computation. 
    \item Performance analysis of the T-BUDS versus UDS \cite{1} (state-of-the-art in lossy utility-driven graph summarization). in terms of running time and memory consumption.
    \item Usefulness analysis of the utility-driven graph summarization framework.
\end{enumerate}

We implemented all algorithms in Java~14 on a single machine with
dual 6 core 2.10~GHz Intel Xeon CPUs, 128~GB RAM and
running Ubuntu 18.04.2 LTS.
Even though our machine had 128~GB we used not more than 20~GB of RAM.  
We used seven web and social graphs from (\url{http://law.di.unimi.it/datasets.php}) varying from moderate size to very large, and we ignored the edge directions and self-loops. 
Table \ref{tbl:2-hop} shows the statistics of these graphs. 
\begin{table}[] 
\centering
{
\begin{tabular}{ c|c|r|r} 
\hline
  Graph & Abbr & Nodes & Edges  \\ \hline
 cnr-2000 &CN& 325,557 & 3,216,152 \\ \hline
 hollywood-2009 &H1& 1,139,905 & 113,891,327 \\ \hline 
hollywood-2011 &H2& 2,180,759 & 228,985,632  \\ \hline
 indochina-2004 &IC& 7,414,866 & 194,109,311\\ \hline
 uk-2002 & U1 & 18,520,486 & 298,113,762 \\ \hline
arabic-2005 &AR& 22,744,080 & 639,999,458 \\ \hline
  uk-2005 & U2 & 39,459,925 & 936,364,282 \\ \hline
\end{tabular}   
 }
 \caption{Summary of datasets \label{tbl:2-hop}}
\end{table}

\subsection{Lossless Case: G-SCIS}

In this section, we evaluate the performance of G-SCIS in terms of (1) Reduction in nodes, (2) running time, and (3) efficiency of triangle enumeration and Pagerank computation. 
For (1) and (2) we compare G-SCIS to SWeG \cite{5}, which is the state-of-the-art in lossless graph summarization. 
\subsubsection{Comparison of G-SCIS to SWeG}
The reduction in nodes (RN) is defined as
$RN = (|V| - |\mathcal{V}|)/|V|$ (c.f.~\cite{1}).
Since SWeG produces also correction graphs for addition/deletion ($C^+,C^-$), RN for SWeG is more precisely computed as 
$RN = (|V| - (|\mathcal{V}| \cup |V(C^+)| \cup |V(C^-|)))/|V|$. 
We ran SWeG for different choices of the number of iterations up to 80 and chose the best RN value obtained.
Figure \ref{fig:G-SCISvsSWeG} shows the comparison between G-SCIS and SWeG in terms of RN and running time. As the figure shows, G-SCIS outperforms SWeG  in term of RN and moreover it is orders of magnitude faster than SWeG. On large graphs like AR and U2, SWeG in not runnable within 100 hours while G-SCIS finishes in around 15 and 23 minutes respectively.

\begin{figure}
\includegraphics[scale=0.75]{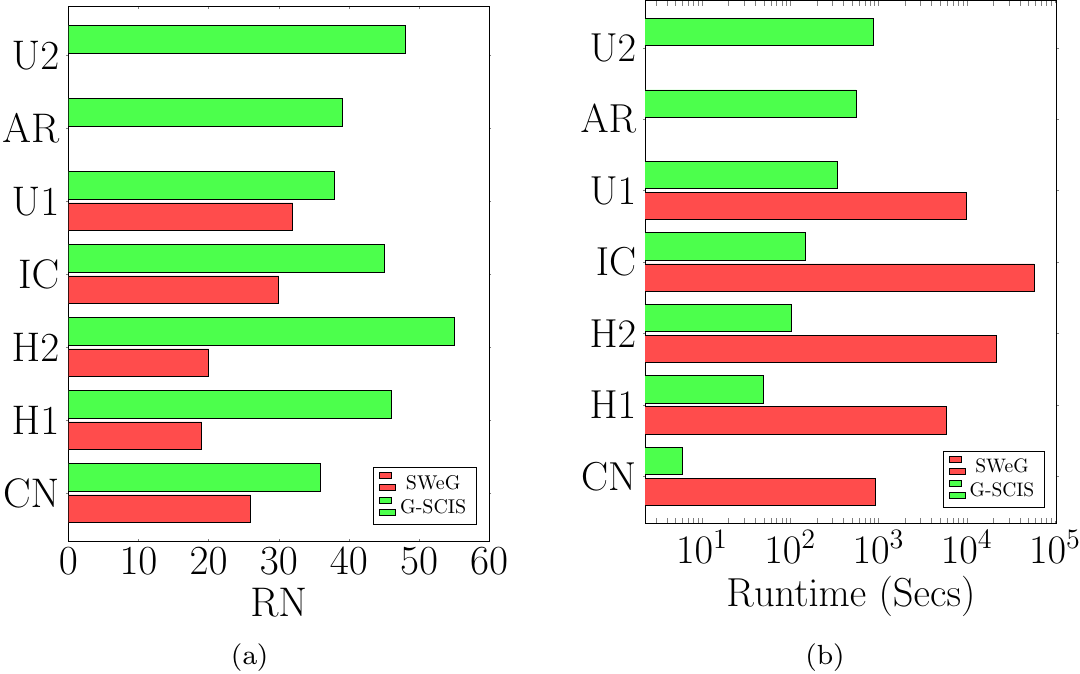}

\caption{Comparison Between G-SCIS and SWeG in terms of node reduction and running time. G-SCIS achieves better reduction than SWeG. Runtime of G-SCIS is orders of magnitude better than SWeG. The latter could not run within 100h for AR and U2. \label{fig:G-SCISvsSWeG}}
\vspace*{-0.4cm}
\end{figure}

\subsubsection{Triangle enumeration and Pagerank computation using G-SCIS summaries} \label{sec:querylossless}


In Figure~\ref{fig:G-SCIS-Query} we show the reduction in runtime for triangle enumeration and Pagerank using G-SCIS summaries versus the runtime of the those algorithms using the original graphs.  
We see a significant reduction in time for both triangle enumeration and Pagerank for all datasets, reaching up to 80\% for IC. We omit results for shortest paths due to space constraints. 
We observe in Figure~\ref{fig:G-SCIS-Query} (left) and (right) a similar order of datasets with some exceptions, such as H2 or U1, for which the order is reversed. We attribute this to the size of the output in triangle enumeration.




%


\begin{figure}
    \centering
\includegraphics[scale=0.75]{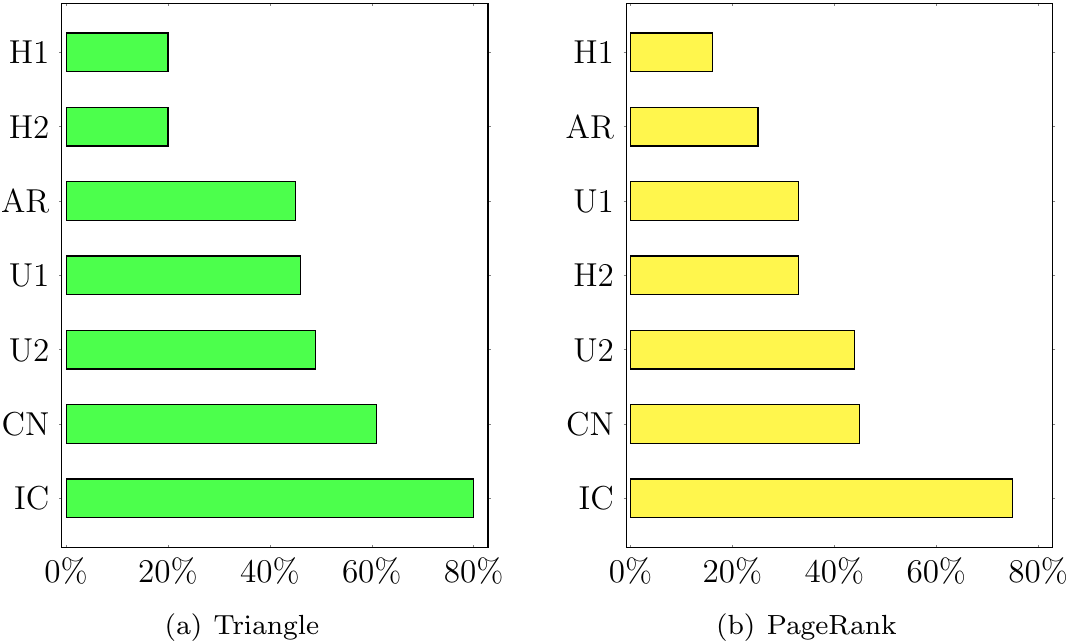}
\caption{Relative improvement on runtime between summary vs original when counting the number of triangles and computing Pagerank. \label{fig:G-SCIS-Query}}
\end{figure}

\subsection{Lossy Case: T-BUDS}

\subsubsection{Performance of T-BUDS}
In this section, the performance of T-BUDS is compared to the performance of UDS in terms of running time and memory usage (Figure~\ref{fig:T-BUDSvsUDS}). 
For our comparison, we set the utility threshold at 0.8.
UDS is quite slow on our moderate and large datasets. Namely, it was not able to complete in reasonable time (100h) for those datasets. 
As such, we provide as input to UDS not the full list of 2-hop pairs as in \cite{1}, but the reduced list from the MST of $G_{2-hop}$. 
This way, we were able to handle with UDS the datasets CN, H1, and H2. However, we still could not have UDS complete for the rest of the datasets.    

Figure~\ref{fig:T-BUDSvsUDS} shows the running time (sec) and memory usage (MB) of T-BUDS and UDS. 
As the figure shows, T-BUDS outperforms UDS in both running time and memory usage by orders of magnitude. 
Moreover, T-BUDS can easily deal with the largest graph, U2, in less than 7 hours.
In contrast UDS takes more than 90 hours on a  moderate graph, such as H2, to produce results. 

\begin{figure}
    \centering
    \includegraphics[scale=0.75]{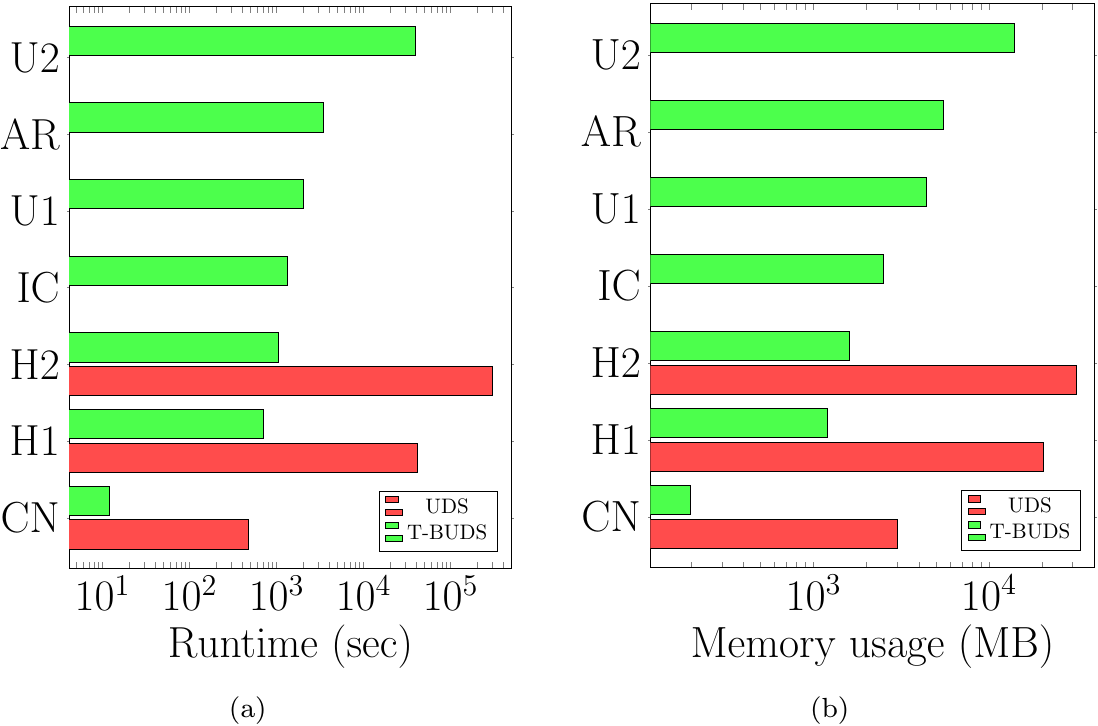}
    \caption{T-BUDS vs UDS in terms of runtime in sec (a) and memory usage in MB (b). $\tau$ is set to 0.8. T-BUDS is orders of magnitude faster than UDS. We provide our MST edge pairs as input for UDS; the original version of UDS could not complete within 100h for all the datasets but CN. With MST as input, UDS still could not complete for IC, U1, AR, and U2. \label{fig:T-BUDSvsUDS}}
\end{figure}

In another experiment we compare the performance of T-BUDS and UDS for varying utility thresholds. 
Figure~\ref{fig:T-BUDS:CN-H09-tau} shows the runtime of the two algorithms on two different graphs CN and H1 in terms of varying utility threshold. 
Having an algorithm that is computationally insensitive to changing the threshold is desirable because it allows the user to conveniently experiment with different values of the threshold.
As shown in the figure, the runtime of T-BUDS remains almost unchanged across different utility thresholds. In contrast, UDS strongly depends on the utility threshold and its runtime grows as the threshold decreases. 

\begin{figure}
    \centering
    \includegraphics[scale=0.75]{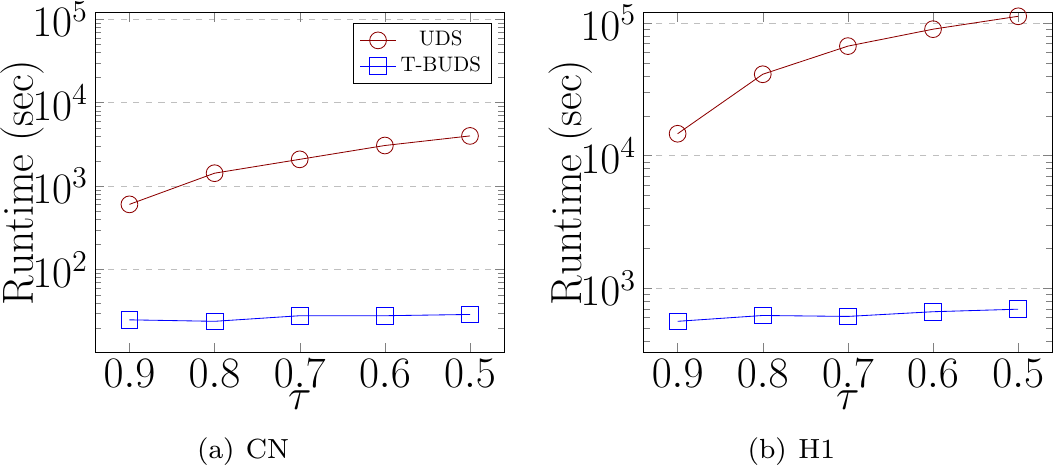}
    \caption{T-BUDS vs UDS for different utility thresholds on CN and H2. T-BUDS is faster by orders of magnitude and, being binary search based, is quite stable as $\tau$ varies.}
    \label{fig:T-BUDS:CN-H09-tau}
\end{figure}

\subsubsection{Usefulness of Utility-Driven Framework}

%
 In this section, we study the performance of T-BUDS towards top-$k$ query answering. 
 To do so, we compute the Pagerank centrality (P) for the nodes, and  assign (normalized) importance score to each edge based on the sum of the importance scores of its two endpoints. We then compute the summary using T-BUDS.
 
 Subsequently, we obtain the top $t\%$ of central nodes in $G$ based on a centrality score (such as Pagerank, Degree, Eigenvector and Betweenness) and check if the corresponding supernode of each such central node in the summary graph is small in size.
 This is desirable because the centrality for a supernode in the summary is divided evenly among the nodes inside it. 
 Towards this, we use the notion of app-utility as defined in \cite{1}.
 Namely, the {\em app-utility} value of a top-$k$ query is as follows. 
 \begin{equation}
  \mbox{app-utility} = \frac{\sum_{v\in V_t}\frac{1}{|S(v)|}}{|V_t|}  
  \label{eq:app-utility}
 \end{equation}
  where $V_t$ is the set of top $t\%$ central nodes, $|V_t| = t\% \times |V|$. The app-utility value is between 0 and 1 and the higher the value, the better the summarization is at capturing the structure of the original graph.  
 $\mbox{app-utility} = 1$ indicates that each central node is a supernode of size 1 and $\mbox{app-utility} < 1$ if there is at least one central node in a supernode of size greater than one, i.e. ``crowded'' with other nodes. 
 
 Table \ref{tbl:app-utility} shows the performance of T-BUDS with varying  $\tau$ and top-$t\%$ central nodes on two graphs CN and H1.
The four columns after RN show the app-utility value of T-BUDS with respect to top-($t\%$) of central nodes on CN graph and the last four columns show the app-utility value of T-BUDS on H1 graph. 
We use four types of top-$t\%$ queries, Pagerank (P),  Degree (D), Eigenvector (E), and Betweenness (B). 
In the table, the first five rows labeled P show the app-utility value of T-BUDS summary with respect to Pagerank query, the next five rows labeled D show the app-utility value of T-BUDS summary with respect to degree query, and so on. 
As can be seen from Table \ref{tbl:app-utility}, T-BUDS performs quite well on top-$t\%$ queries especially for Pagerank and Betweenness centrality measures. 

We compare the performance of T-BUDS 
versus the lossy version of SWeG in terms of the app-utility value for different centrality measures, namely Pagerank, Degree, Eigenvector, and Betweenness. 
In order to fairly compare against SWeG, which does not accept a utility threshold $\tau$ as a parameter, 
we fixed five different $\tau$ values, ${0.5,0.6,0.7,0.8,0.9}$, and calculate the reduction in nodes (RN) for each value when using T-BUDS. Note that $RN = 1 - {|\mathcal{V}|}/{|V|}$.
Then we run SWeG (lossy version) and stop it when each RN value is reached. 
We compute the app-utility value on each summary for the top 20\%, 30\%, 40\% and 50\% of central nodes.

In Figure \ref{fig:App-UtilityPagerank}, we show the relative improvement of T-BUDS over SWeG for two scenarios, $\tau=0.8$ and $\tau=0.6$ for $t=20\%$.
We observe T-BUDS to be significantly better than SWeG. For instance we obtain about 30\% and $50\%$ improvement in app-utility for D for $\tau=0.8$ and $\tau=0.6$.



\begin{table}[]
    \centering
    \caption{T-BUDS: App-utility for top-$k$ queries for Pagerank (P), Degree (D), Eigenvector (E), and Betweenness (B) centralities.}
    \label{tbl:app-utility}
     \resizebox{0.45\textwidth}{!}{
    \begin{tabular}{|c|c|c||c|c|c|c||c|c|c|c|}\hline
     T-BUDS& & & \multicolumn{4}{c}{CN}   & \multicolumn{4}{c|}{H1} \\ \hline
 Centrality &$\tau$& \cellcolor{gray!35}RN & 20\% &30\%&40\%&50\%&20\%&30\%&40\%&50\% \\ \hline
\multirow{5}{*}{\textit{P}} & 0.50&\cellcolor{gray!35}0.58&1.00&1.00&1.00&0.84&1.00&1.00&1.00&0.84 \\ 
&0.60&\cellcolor{gray!35} 0.53&1.00&1.00&1.00&0.94&1.00&1.00&1.00&0.94 \\
&0.70&\cellcolor{gray!35}0.46&1.00&1.00&1.00&1.00&1.00&1.00&1.00&1.00 \\
&0.80&\cellcolor{gray!35}0.38&1.00&1.00&1.00&1.00&1.00&1.00&1.00&1.00 \\
&0.90& \cellcolor{gray!35}0.28&1.00&1.00&1.00&1.00&1.00&1.00&1.00&1.00 \\ \hline\hline
\multirow{5}{*}{\textit{D}}& 0.50&\cellcolor{gray!35}0.58&0.36&0.33&0.38&0.47&0.96&0.82&0.68&0.59 \\ 
&0.60&\cellcolor{gray!35} 0.53&0.40&0.37&0.42&0.51&0.99&0.92&0.81&0.71 \\ 
&0.70&\cellcolor{gray!35}0.46&0.44&0.42&0.47&0.56&1.00&0.97&0.90&0.82 \\ 
&0.80&\cellcolor{gray!35}0.38&0.53&0.53&0.58&0.65&1.00&0.99&0.96&0.90 \\ 
&0.90&\cellcolor{gray!35}0.28&0.63&0.68&0.72&0.77&1.00&0.99&0.99&0.97 \\ \hline\hline
\multirow{5}{*}{\textit{E}}&0.50&\cellcolor{gray!35}0.58&0.48&0.47&0.50&0.44&0.66&0.59&0.52&0.46 \\ 
&0.60&\cellcolor{gray!35}0.53&0.53&0.52&0.55&0.49&0.72&0.66&0.60&0.55 \\
&0.70&\cellcolor{gray!35}0.46&0.58&0.58&0.61&0.54&0.78&0.73&0.67&0.63 \\ 
&0.80&\cellcolor{gray!35}0.38&0.64&0.65&0.69&0.63&0.83&0.80&0.76&0.72 \\ 
&0.90&\cellcolor{gray!35}0.28&0.72&0.74&0.78&0.72&0.88&0.87&0.84&0.82 \\ \hline \hline
\multirow{5}{*}{\textit{B}}& 0.50&\cellcolor{gray!35}0.58&0.60&0.59&0.48&0.44&0.51&0.44&0.43&0.40 \\ 
&0.60&\cellcolor{gray!35}0.53&0.68&0.65&0.53&0.48&0.56&0.51&0.50&0.48 \\ 
&0.70&\cellcolor{gray!35}0.46&0.73&0.71&0.58&0.54&0.62&0.59&0.58&0.57 \\ 
&0.80&\cellcolor{gray!35}0.38&0.80&0.80&0.66&0.62&0.70&0.68&0.67&0.66 \\
&0.90&\cellcolor{gray!35}0.28&0.91&0.90&0.79&0.75&0.80&0.79&0.78&0.77
\\ \hline
\end{tabular}}
\vspace*{-0.2cm}
\end{table}

\begin{figure}
     \centering
    \includegraphics[scale=0.75]{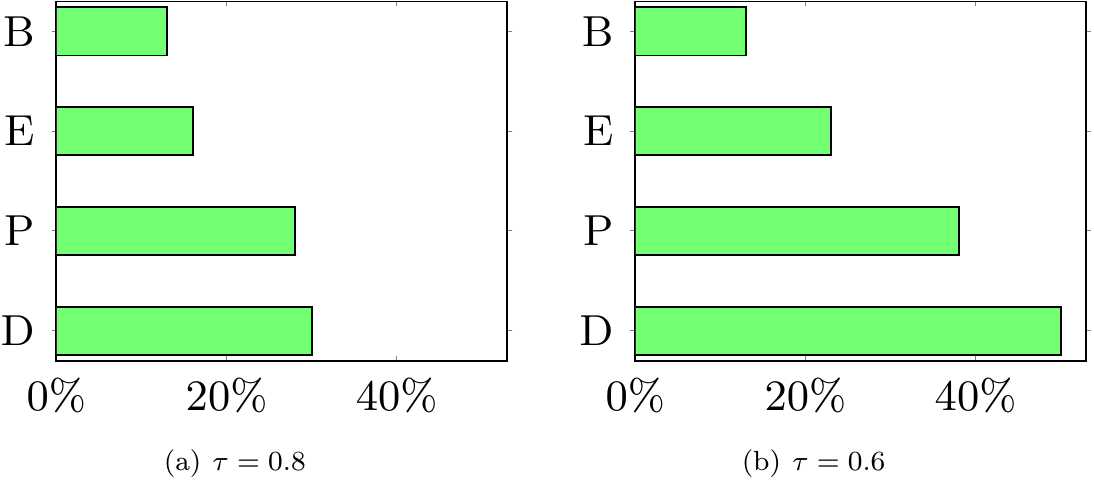}
    \caption{T-BUDS vs SWeG with respect to  app-utility for top-$20\%$  queries on CN. The x axis shows percentages, the y axis shows different queries (D, P, E, B). SWeG lossy was run until it obtained the RN values corresponding to values of $\tau$. The RN values are given in Table~\ref{tbl:app-utility}, i.e. 0.58, 0.53, 0.46, 0.38, 0.28.
    Graph summaries of T-BUDS provide significantly better app-utility than those of SWeG. The difference becomes more pronounced as $\tau$ is lowered.  
    \label{fig:App-UtilityPagerank}}
    \vspace{-0.2cm}
\end{figure}

\section{Related Work \label{sec:RelWorks}}

Graph summarization has been studied in different contexts and we can classify the proposed methodologies into two general categories, grouping and non-grouping. 
The non-grouping category includes sparsification-based methods \cite{shen2006visualanalysis,li2009egocentricinformation,9}
and 
sampling-based methods \cite{hubler2008metropolisalgorithm,jleskovec2006samplinglargegraphs,asmaiya2010samplingcommunitystructure,aahmed2013distributedlargescale,nyan2016previewtables,eliberty2013simple}.
For a more detailed analysis of non-grouping methods, see the survey by Liu et al. \cite{liu2018graph}.

The grouping category of methods is more commonly used for graph summarization and as such has received a lot of attention \cite{1,2,3,4,5,6,7,8,khan2015set}. 
In this category, works such as \cite{2,6} can only produce lossy summarizations optimizing different objectives.
On the other hand, \cite{4,5} are able to generate both lossy and lossless summarizations. 
%
Among works of the grouping category, we discuss the following works \cite{1,4,5,24} that aim to preserve utility and as such are more closely related to our work.


Navlakha et al.~\cite{4} introduced the technique of summarizing the graph by a compact representation containing the summary along with  correction sets. 
Their goal was to minimize the reconstruction error. Liu et. al.~\cite{24} proposed a distributed solution to improve the scalability of the approach in \cite{4}. Recently, Shin et. al. \cite{5}, proposed SWeG, that builds on the work of \cite{4}.  
They used a shingling and minhash based approach to prune the search space for discovering promising candidate pairs. 

 In the work of Kumar and Efstathopoulos \cite{1}, the UDS algorithm was proposed which generates summaries that preserve the utility above a user specified threshold. 
 However, UDS cannot be used for lossless summarization as the summary generated for such a case is the original graph unchanged. 
 Furthermore, for the lossy case, UDS is not scalable to moderate or large graphs.


\section{Conclusions}
In this work, we study utility-driven graph summarization in-depth and made several novel contributions. We present a new, lossless graph summarizer, G-SCIS, that can output the optimal summary, with the smallest number of supernodes, without using correction graphs as in previous approaches. We show the versatility of the G-SCIS summary using popular queries such as enumerating triangles, estimating Pagerank and computing shortest paths. 

We design a scalable, lossy summarization algorithm, T-BUDS. Two key insights leading to the scalability of T-BUDS are the use of MST of the two-hop graph combined with binary search over the MST edges. We demonstrate the effectiveness of T-BUDS towards answering top-$t\%$ queries based on popular centrality measures such as Pagerank, degree, eigenvector and betweenness.

 \newpage
 \balance
 \bibliography{main}

\begin{thebibliography}{10}

\bibitem{facebook}
Facebook by the numbers: Stats, demographics \& fun facts.
\newblock \url{https://www.omnicoreagency.com/facebook-statistics}.
\newblock Accessed: 2020-05-23.

\bibitem{websites}
How many websites are there around the world? [2020].
\newblock \url{https://www.millforbusiness.com/how-many-websites-are-there}.
\newblock Accessed: 2020-05-23.

\bibitem{weibo}
Number of sina weibo users in china from 2017 to 2021.
\newblock
  \url{https://www.statista.com/statistics/941456/china-number-of-sina-weibo-users}.
\newblock Accessed: 2020-05-23.

\bibitem{twitter}
Twitter by the numbers: Stats, demographics \& fun facts.
\newblock \url{https://www.omnicoreagency.com/twitter-statistics}.
\newblock Accessed: 2020-05-23.

\bibitem{aahmed2013distributedlargescale}
{\sc Ahmed, A., Shervashidze, N., Narayanamurthy, S., Josifovski, V., and
  Smola, A.~J.}
\newblock Distributed large-scale natural graph factorization.
\newblock In {\em Proceedings of the 22nd international conference on World
  Wide Web\/} (2013), pp.~37--48.

\bibitem{apostolico2009graph}
{\sc Apostolico, A., and Drovandi, G.}
\newblock Graph compression by bfs.
\newblock {\em Algorithms 2}, 3 (2009), 1031--1044.

\bibitem{boldi2004webgraph}
{\sc Boldi, P., and Vigna, S.}
\newblock The webgraph framework i: compression techniques.
\newblock In {\em Proceedings of the 13th international conference on World
  Wide Web\/} (2004), pp.~595--602.

\bibitem{22}
{\sc Cook, D.~J., and Holder, L.~B.}
\newblock Substructure discovery using minimum description length and
  background knowledge.
\newblock {\em Journal of Artificial Intelligence Research 1\/} (1993),
  231--255.

\bibitem{19}
{\sc Dunne, C., and Shneiderman, B.}
\newblock Motif simplification: improving network visualization readability
  with fan, connector, and clique glyphs.
\newblock In {\em Proceedings of the SIGCHI Conference on Human Factors in
  Computing Systems\/} (2013), pp.~3247--3256.

\bibitem{3}
{\sc Fan, W., Li, J., Wang, X., and Wu, Y.}
\newblock Query preserving graph compression.
\newblock In {\em Proceedings of the 38th ACM SIGMOD International Conference
  on Management of Data\/} (2012), pp.~157--168.

\bibitem{7}
{\sc Gou, X., Zou, L., Zhao, C., and Yang, T.}
\newblock Fast and accurate graph stream summarization.
\newblock In {\em Proceedings of the 35th IEEE International Conference on Data
  Engineering (ICDE)\/} (2019), pp.~1118--1129.

\bibitem{23}
{\sc Hay, M., Miklau, G., Jensen, D., Towsley, D., and Weis, P.}
\newblock Resisting structural re-identification in anonymized social networks.
\newblock {\em Proceedings of the VLDB Endowment 1}, 1 (2008), 102--114.

\bibitem{hopcraft1973set}
{\sc Hopcroft, J.~E., and Ullman, J.~D.}
\newblock Set merging algorithms.
\newblock {\em SIAM Journal on Computing 2}, 4 (1973), 294--303.

\bibitem{hubler2008metropolisalgorithm}
{\sc H{\"u}bler, C., Kriegel, H.-P., Borgwardt, K., and Ghahramani, Z.}
\newblock Metropolis algorithms for representative subgraph sampling.
\newblock In {\em Proceedings of the 8th IEEE International Conference on Data
  Mining (ICDM)\/} (2008), pp.~283--292.

\bibitem{khan2015set}
{\sc Khan, K.~U., Nawaz, W., and Lee, Y.-K.}
\newblock Set-based approximate approach for lossless graph summarization.
\newblock {\em Computing 97}, 12 (2015), 1185--1207.

\bibitem{21}
{\sc Koutra, D., Kang, U., Vreeken, J., and Faloutsos, C.}
\newblock Summarizing and understanding large graphs.
\newblock {\em Statistical Analysis and Data Mining: The ASA Data Science
  Journal 8}, 3 (2015), 183--202.

\bibitem{kruskal1956shortest}
{\sc Kruskal, J.~B.}
\newblock On the shortest spanning subtree of a graph and the traveling
  salesman problem.
\newblock {\em Proceedings of the American Mathematical society 7}, 1 (1956),
  48--50.

\bibitem{1}
{\sc Kumar, K.~A., and Efstathopoulos, P.}
\newblock Utility-driven graph summarization.
\newblock {\em Proceedings of the VLDB Endowment 12}, 4 (2018), 335--347.

\bibitem{2}
{\sc LeFevre, K., and Terzi, E.}
\newblock Grass: Graph structure summarization.
\newblock In {\em Proceedings of the 10th SIAM International Conference on Data
  Mining (SDM)\/} (2010), pp.~454--465.

\bibitem{jleskovec2006samplinglargegraphs}
{\sc Leskovec, J., and Faloutsos, C.}
\newblock Sampling from large graphs.
\newblock In {\em Proceedings of the 12th ACM SIGKDD international conference
  on Knowledge discovery and data mining\/} (2006), pp.~631--636.

\bibitem{20}
{\sc Li, C., Baciu, G., and Wang, Y.}
\newblock Modulgraph: modularity-based visualization of massive graphs.
\newblock In {\em Proceedings of the SIGGRAPH Asia 2015 Visualization in High
  Performance Computing\/} (2015), pp.~1--4.

\bibitem{li2009egocentricinformation}
{\sc Li, C.-T., and Lin, S.-D.}
\newblock Egocentric information abstraction for heterogeneous social networks.
\newblock In {\em Proceedings of the 1st International Conference on Advances
  in Social Network Analysis and Mining\/} (2009), pp.~255--260.

\bibitem{eliberty2013simple}
{\sc Liberty, E.}
\newblock Simple and deterministic matrix sketching.
\newblock In {\em Proceedings of the 19th ACM SIGKDD international conference
  on Knowledge discovery and data mining\/} (2013), pp.~581--588.

\bibitem{24}
{\sc Liu, X., Tian, Y., He, Q., Lee, W.-C., and McPherson, J.}
\newblock Distributed graph summarization.
\newblock In {\em Proceedings of the 23rd ACM International Conference on
  Conference on Information and Knowledge Management\/} (2014), pp.~799--808.

\bibitem{liu2018graph}
{\sc Liu, Y., Safavi, T., Dighe, A., and Koutra, D.}
\newblock Graph summarization methods and applications: A survey.
\newblock {\em ACM Computing Surveys (CSUR) 51}, 3 (2018), 1--34.

\bibitem{9}
{\sc Maccioni, A., and Abadi, D.~J.}
\newblock Scalable pattern matching over compressed graphs via dedensification.
\newblock In {\em Proceedings of the 22nd ACM SIGKDD International Conference
  on Knowledge Discovery and Data Mining\/} (2016), pp.~1755--1764.

\bibitem{asmaiya2010samplingcommunitystructure}
{\sc Maiya, A.~S., and Berger-Wolf, T.~Y.}
\newblock Sampling community structure.
\newblock In {\em Proceedings of the 19th international conference on World
  wide web\/} (2010), pp.~701--710.

\bibitem{4}
{\sc Navlakha, S., Rastogi, R., and Shrivastava, N.}
\newblock Graph summarization with bounded error.
\newblock In {\em Proceedings of the 34th ACM SIGMOD international conference
  on Management of data\/} (2008), pp.~419--432.

\bibitem{Page1999ThePC}
{\sc Page, L., Brin, S., Motwani, R., and Winograd, T.}
\newblock The pagerank citation ranking: Bringing order to the web.
\newblock Tech. rep., Stanford InfoLab, 1999.

\bibitem{prim1957shortest}
{\sc Prim, R.~C.}
\newblock Shortest connection networks and some generalizations.
\newblock {\em The Bell System Technical Journal 36}, 6 (1957), 1389--1401.

\bibitem{6}
{\sc Riondato, M., Garc{\'\i}a-Soriano, D., and Bonchi, F.}
\newblock Graph summarization with quality guarantees.
\newblock In {\em Proceedings of the 14th IEEE International Conference on Data
  Mining (ICDM)\/} (2014), pp.~947--952.

\bibitem{rossi2018graphzip}
{\sc Rossi, R.~A., and Zhou, R.}
\newblock Graphzip: a clique-based sparse graph compression method.
\newblock {\em Journal of Big Data 5}, 1 (2018), 10.

\bibitem{santoso2019triad}
{\sc Santoso, Y., Thomo, A., Srinivasan, V., and Chester, S.}
\newblock Triad enumeration at trillion-scale using a single commodity machine.
\newblock In {\em Proceedings of the 22nd International Conference on Extending
  Database Technology (EDBT)\/} (2019), pp.~718--721.

\bibitem{shah2017summarizing}
{\sc Shah, N., Koutra, D., Jin, L., Zou, T., Gallagher, B., and Faloutsos, C.}
\newblock On summarizing large-scale dynamic graphs.
\newblock {\em IEEE Data Eng. Bull. 40}, 3 (2017), 75--88.

\bibitem{5}
{\sc Shin, K., Ghoting, A., Kim, M., and Raghavan, H.}
\newblock Sweg: Lossless and lossy summarization of web-scale graphs.
\newblock In {\em Proceedings of the 28th international conference on World
  Wide Web\/} (2019), pp.~1679--1690.

\bibitem{spielman2011graph}
{\sc Spielman, D.~A., and Srivastava, N.}
\newblock Graph sparsification by effective resistances.
\newblock {\em SIAM Journal on Computing 40}, 6 (2011), 1913--1926.

\bibitem{8}
{\sc Tian, Y., Hankins, R.~A., and Patel, J.~M.}
\newblock Efficient aggregation for graph summarization.
\newblock In {\em Proceedings of the 34th ACM SIGMOD international conference
  on Management of data\/} (2008), pp.~567--580.

\bibitem{nyan2016previewtables}
{\sc Yan, N., Hasani, S., Asudeh, A., and Li, C.}
\newblock Generating preview tables for entity graphs.
\newblock In {\em Proceedings of the 2016 International Conference on
  Management of Data\/} (2016), pp.~1797--1811.

\bibitem{shen2006visualanalysis}
{\sc {Zeqian Shen}, {Kwan-Liu Ma}, and {Eliassi-Rad}, T.}
\newblock Visual analysis of large heterogeneous social networks by semantic
  and structural abstraction.
\newblock {\em IEEE Transactions on Visualization and Computer Graphics 12}, 6
  (2006), 1427--1439.

\end{thebibliography}
 \bibliographystyle{acm}
\end{document}